\documentclass[11pt]{article}
\usepackage[T1]{fontenc}
\usepackage{lmodern}
\usepackage[margin=1in]{geometry}
\usepackage{amsmath,amssymb,amsthm,mathtools}
\usepackage{booktabs,array,microtype}
\usepackage{graphicx}
\PassOptionsToPackage{hyphens}{url}
\usepackage[colorlinks=true,linkcolor=blue,citecolor=blue,urlcolor=blue]{hyperref}
\newcommand*{\paperauthor}{Zhe Zhang, Yi Liu and Dan Yang}
\newcommand*{\paperauthorblock}{%
\begin{minipage}{\dimexpr\textwidth-2\tabcolsep\relax}
\centering
{\large Zhe Zhang$^{1}$, Yi Liu$^{1}$,
Dan Yang$^{2}$}\\[0.65em]
{\small $^{1}$UniDT, Shanghai, China\\[0.35em]
$^{2}$Yangtze Delta Region Institute (Quzhou),\\
University of Electronic Science and Technology of China, Quzhou, China}\\[0.65em]
{\small Zhe Zhang: \href{mailto:zz940521@gmail.com}{\texttt{zz940521@gmail.com}};\quad
Yi Liu: \href{mailto:yi.liu@unidt.com}{\texttt{yi.liu@unidt.com}}\\[0.25em]
Dan Yang: \href{mailto:yangdan0504@csj.uestc.edu.cn}{\texttt{yangdan0504@csj.uestc.edu.cn}}}
\end{minipage}}

\hypersetup{pdftitle={Two-colour crossing minimization in complete geometric graphs is NP-hard},pdfauthor={\paperauthor},pdfsubject={Complexity of two-colour crossing minimization},pdfkeywords={geometric graph, monochromatic crossings, MaxCut, NP-completeness}}
\newtheorem{theorem}{Theorem}[section]
\newtheorem{lemma}[theorem]{Lemma}
\newtheorem{proposition}[theorem]{Proposition}
\newtheorem{corollary}[theorem]{Corollary}
\theoremstyle{definition}
\newtheorem{definition}[theorem]{Definition}
\theoremstyle{remark}

\numberwithin{equation}{section}

\newcommand{\Q}{\mathbb Q}
\newcommand{\Z}{\mathbb Z}

\newcommand{\1}{\mathbf 1}
\newcommand{\eps}{\varepsilon}
\DeclareMathOperator{\OPT}{OPT}
\DeclareMathOperator{\MaxCut}{MaxCut}
\DeclareMathOperator{\sgn}{sgn}
\DeclareMathOperator{\orient}{or}

\DeclarePairedDelimiter{\floor}{\lfloor}{\rfloor}
\DeclarePairedDelimiter{\ceil}{\lceil}{\rceil}

\title{Two-colour crossing minimization\protect\\ in complete geometric graphs is NP-hard}
\author{\paperauthorblock}
\date{5 November 2026}
\begin{document}
\maketitle
\begin{abstract}
Given a fixed straight-line drawing of a complete graph, we seek a
two-colouring of its edges that minimizes monochromatic crossing pairs.
We prove that the threshold decision problem is NP-complete, even for
points in general position with integer coordinates. Thus exact
minimization is NP-hard, resolving the complexity question posed at the
Crossing Number Workshop in 2019. We reduce from MaxCut on permutation
graphs by representing source vertices with classes of geometric edges,
whose crossings encode source adjacencies. The target complete graph also
contains every other edge between its points, and their freely chosen
colours can affect the encoded cost. To control these interactions, we
construct a large comparison grid with two stable colour phases and use
program points to turn this reference into colour preferences on other
edges. The optimum differs from a computable baseline plus the
amplified minimum number of uncut source edges by an error small enough to distinguish consecutive
source values. The reduction is deterministic, with polynomially many points and
polynomial coordinate bit length.
\end{abstract}
\section{Introduction}\label{sec:introduction}

We study the complexity of minimizing monochromatic crossings over all
edge two-colourings of a fixed complete geometric graph. This models
placing the two colour classes on separate layers with the same vertex
positions.

The complexity of finding an optimal two-colouring of a fixed straight-line
drawing of $K_N$ was posed as Question~4 at the Crossing Number Workshop
in 2019 \cite[Section~1.4]{CRNW2019} and as the first open problem in
Section~5 of Aichholzer et al.\ \cite{Aichholzer2019}. It was subsequently
listed in Dagstuhl Seminar~21293 \cite[Section~4.12]{Dagstuhl2021} and at
Geometric Graphs Week~2022 \cite[Problem~11]{GGWeek2022}.
We prove that the threshold decision version is NP-complete, and hence
that exact minimization is NP-hard. All edges may receive either colour;
the input has no edge weights or prescribed colours.

\begin{theorem}\label{thm:main}
The decision problem in Definition~\ref{def:problem} is NP-complete under
deterministic polynomial-time many-one reductions. Consequently, computing
$\OPT(\mathcal P)$ exactly is NP-hard. The decision problem remains NP-complete
when all coordinates are integers encoded in binary.
\end{theorem}

Aichholzer, Mulzer, Schnider and Vogtenhuber \cite{Aichholzer2018} proved
NP-completeness of MaxCut on segment intersection graphs. Completing
a segment family introduces new colour variables and costs that can depend
on the original colours (Proposition~\ref{prop:completion-obstruction}). Their result
therefore does not settle the complete-graph problem.

The contrast is already visible in convex position. Two-colour crossing
minimization for an arbitrary graph drawn on a convex point set is
NP-hard, by the fixed-order two-page result of Masuda et al.\
\cite{Masuda1990}. For a complete graph, however, the optimum is the known
two-page crossing number of $K_N$ \cite{Abrego2013}. Our reduction uses
this exact value for two convex background components.

We reduce from MaxCut on permutation graphs, whose NP-completeness proof
supplies a permutation representation \cite{Figueiredo2023}. Classes of
chords encode source vertices and amplify the number of uncut source
edges. Program points induce costs on these chords using a comparison
grid as a colour reference, while the remaining edges contribute a
computable baseline and an error smaller than one amplified source unit.

Aichholzer et al.\ \cite{Aichholzer2019}
give two-colourings of a double-chain bipartite drawing with
asymptotic crossing ratio at most $1/3$ (Appendix~B, Theorem~6).
We prove the corresponding
continuous minimum and two-phase stability, then obtain localization for
finite rectangles under arbitrary bounded unary perturbations. The
stability exponent is not optimized.

Other results concern extremal or approximate versions of the objective.
Bounds over all drawings and duplication constructions were developed in
\cite{Aichholzer2019}; extensions to dense graphs, random points and
several colours appear in \cite{FabilaMonroy2025,Cabello2024,Hahn2025}.
For a fixed complete geometric graph, the crossing graph is dense, and
the MaxCut scheme of Arora, Karger and Karpinski \cite{Arora1999} gives
additive approximation on the scale of $N^4$, as observed in
\cite{Hahn2025}. These results do not determine the exact optimum.

Lean proof code and reproducible exact computations are available in the
accompanying repository \cite{ProofRepository}.

Section~\ref{sec:overview} outlines the construction. We then prove the grid
estimates in Section~\ref{sec:stability}, realize the crossing relations in
Section~\ref{sec:geometry}, and analyze the core in Section~\ref{sec:core}.
Section~\ref{sec:reduction} completes the proof of Theorem~\ref{thm:main}.

\subsection{Preliminaries}

\begin{definition}\label{def:problem}
The input to \emph{complete geometric two-colour crossing minimization}
is a list of pairwise distinct points $\mathcal P\in(\Q^2)^N$ with no three
collinear, together with an integer $k_{\rm tgt}$ in the decision version.
Coordinates and $k_{\rm tgt}$ are encoded in binary. Write
\[
 E(\mathcal P)=\{ij:0\leq i<j<N\},\qquad
 M(\mathcal P,\chi)=\sum_{\{e,f\}\subseteq E(\mathcal P)}
       \1[e\text{ and }f\text{ cross properly}]\,\1[\chi(e)=\chi(f)]
\]
for $\chi:E(\mathcal P)\to\{0,1\}$, and put
$\OPT(\mathcal P)=\min_\chi M(\mathcal P,\chi)$.
The decision question is whether $\OPT(\mathcal P)\leq k_{\rm tgt}$.
\end{definition}

A proper crossing is an intersection of two open segments with four
distinct endpoints. We write $e\sim f$ when edges $e$ and $f$ cross properly.
Several crossing pairs may have the same geometric
intersection point; each pair is counted separately. The colouring is
unrestricted: it is not a proper edge colouring in the graph-theoretic
sense.

The crossing graph $\Gamma(\mathcal P)$ has one vertex for each geometric
edge and an adjacency for each proper crossing. Writing
$\operatorname{cr}(\mathcal P)=|E(\Gamma(\mathcal P))|$, we have
\begin{equation}\label{eq:crossing-maxcut}
 \OPT(\mathcal P)=\operatorname{cr}(\mathcal P)-\MaxCut(\Gamma(\mathcal P)).
\end{equation}
Indeed, every crossing is either monochromatic or belongs to the cut
induced by the two edge colours. Throughout, graph edges and crossing pairs
are unordered unless an order is explicitly specified.

A represented permutation graph $G$ on $[n]=\{1,\ldots,n\}$ is specified,
after relabelling one order, by a permutation $\pi$, with $ij\in E(G)$
if and only if $i<j$ and $\pi(i)>\pi(j)$. A permutation representation
can also be found from the abstract graph in deterministic
$O(|V(G)|+|E(G)|)$ time \cite[pp.~212--213]{McConnell1999}. Put
\[
 g(G)=\min_{x:V(G)\to\{0,1\}}\sum_{uv\in E(G)}\1[x(u)=x(v)]
     =|E(G)|-\MaxCut(G).
\]
A source threshold $k_{\rm src}$ is thus equivalent to
$g(G)\le t$, where $t=|E(G)|-k_{\rm src}$.
We use $N$ for the number of target points and $N_g$ for the number of
comparison-grid vertices, each of which is a target edge.

For nonnegative integers $v$, define
\[
 Z(v)=\frac14\left\lfloor\frac v2\right\rfloor
                  \left\lfloor\frac{v-1}{2}\right\rfloor
                  \left\lfloor\frac{v-2}{2}\right\rfloor
                  \left\lfloor\frac{v-3}{2}\right\rfloor
 \quad(v\geq3),\qquad Z(0)=Z(1)=Z(2)=0.
\]
The theorem of \'{A}brego et al.\ \cite[Theorem~7]{Abrego2013} says that
$Z(v)$ is the optimum for a convex complete geometric graph: a same-page
pair of edges in a two-page drawing crosses exactly when its endpoints
alternate in the linear order. Boundary edges do not affect this count.

\begin{proposition}\label{prop:np}
The problem in Definition~\ref{def:problem} belongs to NP. Rational and
integer coordinate versions are polynomial-time equivalent.
\end{proposition}
\begin{proof}
Let $L$ be the input bit length. Normalize the coordinate denominators
to be positive and let $D$ be their product. Multiplying all coordinates by $D$ gives
integers of $O(L)$ bits, preserves orientation signs, and takes polynomial
time.

Pairwise distinctness and absence of collinear triples can be checked
by examining point pairs and the determinants
\[
 \orient(a,b,c)=(b_x-a_x)(c_y-a_y)-(b_y-a_y)(c_x-a_x).
\]
These tests determine whether the point set satisfies the conditions in
Definition~\ref{def:problem}. A certificate consists of the
$\binom N2$ edge colours. For every unordered pair $ab,cd$ with four
distinct endpoints, a proper crossing is equivalent to
\[
 \orient(a,b,c)\orient(a,b,d)<0,
 \qquad \orient(c,d,a)\orient(c,d,b)<0.
\]
The verifier accepts precisely when the point set satisfies these
conditions and the monochromatic crossing count is at most
$k_{\rm tgt}$. The $O(N^4)$ tests use
integers of polynomial bit length, so verification takes polynomial time.
Scaling reduces rational inputs to integer inputs; the converse inclusion
is immediate.
\end{proof}

\subsection{A direct-completion obstruction}\label{subsec:completion}

\begin{proposition}\label{prop:completion-obstruction}
Completing a segment family to the complete geometric graph on its
endpoints can produce an optimum whose restriction to the original
family is not optimal. The minimum contribution of the added edges can
also depend on the prescribed source colours.
\end{proposition}
\begin{proof}
Let $p_i=(i,i^2)$, $0\le i<6$, and abbreviate $p_ip_j$ by $ij$.
The source family $\{01,24,35\}$ has only one crossing, between $24$ and
$35$, and therefore has optimum zero.

The complete drawing is convex, so its optimum is $Z(6)=3$ by
\'{A}brego et al.\ \cite[Theorem~7]{Abrego2013}.

Colour $03,04,13,14$ blue and every other edge red. The monochromatic
pairs are exactly $(02,15),(03,14),(24,35)$, attaining this lower bound.
Its source restriction has cost one and is not optimal. Conditional on
$24,35$ having the same colour, the minimum added cost is thus two.

Instead colour $14,15,24,25$ blue and every other edge red. The
monochromatic pairs are $(02,13),(04,35),(14,25)$, again attaining three.
Now $24,35$ differ, the source cost is zero, and the minimum added cost is
three. Hence the optimum added cost depends on the source colours.
\end{proof}

\section{Overview of the construction}\label{sec:overview}

We encode a source cut by colours of selected edges, then control the
other edges required by completeness. Line styles in the diagrams
distinguish edge roles.

\subsection{From a source cut to selected chords}

Placing the endpoint orders of a permutation representation in early and
late blocks on a convex chain, with the second order reversed, gives
chords that cross exactly for adjacent source vertices. Their monochromatic
crossings count the uncut source edges.

To amplify this cost, replace each chord by $q$ nested macro chords.
Each macro endpoint becomes a cluster of $W$ points; all $W^2$ edges
between paired clusters form one bundle. One source vertex is thus
represented by $q$ bundles, or $qW^2$ selected edges
(Figure~\ref{fig:source}). Adjacent source vertices give
$(qW^2)^2=\Delta$ crossing pairs between their classes. If each class is
monochromatic, the cost is $\Delta$ when their colours agree and zero
otherwise.

Identical inter-class neighbours give a class-monochromatic optimum of
value $\Delta g(G)$ for this term (Lemma~\ref{lem:source-amplification}).
Within-bundle crossings and certain short-edge interactions, set aside
in the analysis, cost at most $2SW^4$ to restore (Lemma~\ref{lem:local-count}).

\begin{figure}[!hb]
\centering
\includegraphics{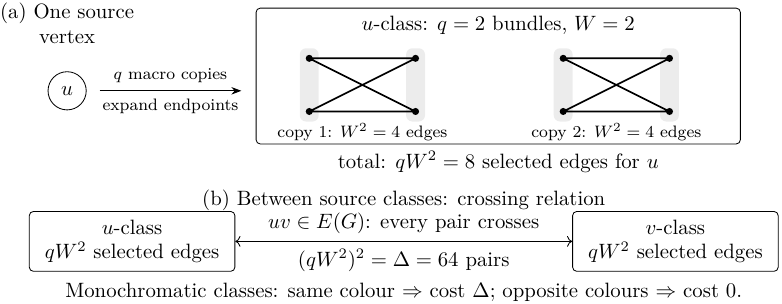}
\caption{Source expansion at illustrative sizes.
Each shaded strip in (a) is a $W$-point cluster. Both bundles contain all
$W^2$ edges; these are separate local views, not their global positions. In (b), the connector denotes pairwise
crossing, not a target edge; nonadjacent classes have no crossings.}
\label{fig:source}
\end{figure}

\subsection{Controlling the completion edges}

The original chain $P$ contains the endpoint clusters and a middle block
of dummy points. The $PP$ edges form the core $C$. Its selected edges
encode the cut, but the other core edges can alter its cost.
We add a chain of frame points $A$ and program
points $R$ to control them. The target is complete on $P\cup A\cup R$;
notation such as $AP$ names an edge family (Figure~\ref{fig:layout}).

The $AP$ edges, denoted by $X$, form a rectangular comparison grid with two complementary
reference colourings, called phases; position $(i,j)$ represents the
edge $A_iP_j$. In a free optimum the grid stays
close to one phase, even with all external costs. Complementing all
colours exchanges the phases and preserves the cost.

Geometric separation prevents $AP$--$PP$ crossings. The $RA$ edges,
called program fans and denoted by $T_f$, connect the systems by crossing
grid and core edges.
Once a fan's colour is fixed, each core edge it crosses pays one in that
colour and zero in the other. This induces a unary cost depending on
that edge's colour.

We analyze the core with the fans fixed to their canonical responses to
the grid. Ordinary programs then force a known colouring on the palette
edges $C_{\rm pal}$. Rank programs compensate for the palette's varying
bias across each selected bundle; the ordinary programs are calibrated
so that only a small residual bias remains there. Short edges and
dummy--dummy edges remain free, with their minima accounted for separately.
The final argument bounds the cost of fixing the fans while keeping the
actual grid colouring unchanged.

\begin{figure}[tbp]
\centering
\begin{minipage}{.48\linewidth}\centering
\includegraphics{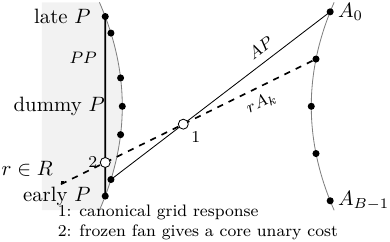}
\end{minipage}\hfill
\begin{minipage}{.49\linewidth}
\small
\begin{tabular}{@{}l>{\raggedright\arraybackslash}p{.64\linewidth}@{}}
\toprule
Edge class & Role in the proof\\\midrule
$X=AP$ & Stable grid: $A_iP_j\leftrightarrow(i,j)$\\
$T_f=RA$ & Canonical grid response, then frozen core costs\\
$C=PP$ & Selected bundles, palette, short and dummy edges\\
$Y=RP\cup RR$ & Restored within the bound $R_Y$\\
$I=AA$ & Isolated convex component of value $Z(B)$\\\bottomrule
\end{tabular}
\end{minipage}
\caption{Layout and edge classes (not to scale). The shaded convex region
contains $P,R$; $AP$ interiors lie outside. The same fan crosses the canonical
grid at 1; freezing its response gives a core unary cost at 2. Open circles
mark crossings; the diagram shows representative edges.}
\label{fig:layout}
\end{figure}

All remaining edge families also enter the cost. The $AA$ edges, denoted by $I$, form an
isolated convex component. Separation excludes crossings of $AP$ with
$RP$ or $RR$, and a uniform bound controls all other interactions of
$Y=RP\cup RR$. The resulting estimate is
\begin{equation}\label{eq:main-sandwich}
 K_{\rm base}+\Delta g(G)-E_-
 \le \OPT(\mathcal P)\le K_{\rm base}+\Delta g(G)+E_+,
 \qquad 4(E_-+E_+)<\Delta.
\end{equation}
The baseline uses reference colours, unary corrections and known convex
optima. Proposition~\ref{prop:decoupling} proves the bounds by freezing fans
and restoring $RP\cup RR$. Integrality of $g(G)$ then transfers the source
threshold to one complete drawing.

\subsection{Parameters and their roles}

For $n\ge1$ fix the following integers:
\begin{equation}\label{eq:parameters}
\begin{gathered}
 q=16(n+1),\quad S=nq,\quad d=S+1,\quad H=2^{4200}d^{14},\\
 W=H^4,\quad B=H^5,\quad D_m=512SW,\quad o=2SW+D_m=514SW,\\
 T=2^{13}o^2,\qquad \Lambda_{\rm pal}=2^{13}dT,\qquad
 \Delta=q^2W^4.
\end{gathered}
\end{equation}
There are $S$ selected macro chords, $o$ original points and $B$ frames.
If $s$ is the number of program points, the complete drawing has
$N=B+o+s$ points, whereas the comparison grid has $N_g=Bo$ vertices.
The symbol $n$ in the abstract grid $J_{m,n}$ is a local dimension;
in the construction those dimensions are $B,o$.

The cluster size $W$ amplifies the source signal, while $q$ controls the
relative local restoration cost. The scale $H$ makes the grid perturbations,
freezing cost and nonfan interactions small relative to that signal.
The frame count $B$ controls the grid strength and program weights,
and the dummy block makes both selected rank slopes negative.
The spacing $T$ separates palette preferences, and $\Lambda_{\rm pal}$
dominates adverse core interactions.
Section~\ref{sec:reduction} verifies these scale requirements. The constants
are not optimized, and the resulting instance sizes remain impractical.

\section{Stability of the comparison grid}
\label{sec:stability}

In the comparison grid (Figure~\ref{fig:grid}), two positions are adjacent
when both coordinate orders strictly agree (Lemma~\ref{lem:double-chain}).
Fixing outside colours induces unary costs $U_{ij}(c)$: a fixed fan
$rA_k$ contributes $\1[c=t_{rk}]$ at each grid edge that it crosses.

For a grid colouring $c$, the cellwise sign function is $f_c=2c-1$.
The normalized grid cost $\Phi_J(c)/N_g^2$ approximates
$\mathcal E(f_c)$ below, with a mesh error. Near-minimal energy gives nearly
threshold rows and columns, then a profile close to one of two phases.
Single-vertex optimality in a finite grid confines
disagreements to a diagonal strip and two corners.
Theorem~\ref{thm:grid} also gives a lower bound for any colouring with
that localized support. We use this bound after changing the program
colours while keeping the grid fixed.

\begin{figure}[tbp]
\centering
\includegraphics{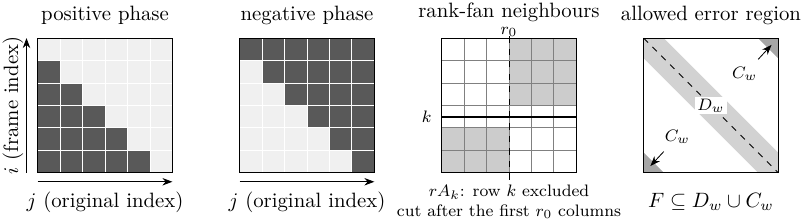}
\caption{The comparison grid represents $AP$ edges, not geometric points.
In the first two panels dark cells have colour $1$ and light cells colour
$0$; the negative phase is the exact complement, including diagonal ties.
In the third panel shading denotes neighbours of a rank fan, independently
of colour: $i<k,j<r_0$ or $i>k,j\ge r_0$.
The last panel shows the diagonal strip and two corner regions containing
all errors $F$ of a minimizing colouring relative to a phase. The displayed
mesh and width are illustrative, not the quantitative theorem's scales.}
\label{fig:grid}
\end{figure}

\subsection{The continuous comparison energy}

All integrals in this section use Lebesgue measure.  Functions are Borel
measurable; the statements also hold for Lebesgue-measurable functions after
replacement on a null set.  For a sign function
$f:[0,1]^2\to\{-1,1\}$, put
\begin{equation}
 \mathcal E(f)=
 \iiiint_{\substack{0\le u<u'\le1\\0\le v<v'\le1}}
 \mathbf 1\{f(u,v)=f(u',v')\}\,du\,dv\,du'\,dv'.
 \label{eq:comparison-energy}
\end{equation}
Each comparable unordered pair occurs once.  Define
\[
 f_*(u,v)=
 \begin{cases}1,&u+v<1,\\-1,&u+v\ge1.\end{cases}
\]
The diagonal convention has no effect on the continuous integral, but we
retain this strict inequality when passing to finite grids.

\begin{theorem}[Continuous two-phase stability]
\label{thm:continuous}
Every measurable sign function on the unit square has
$\mathcal E(f)\ge1/12$.  For every $\varepsilon\ge0$, if
$\mathcal E(f)\le1/12+\varepsilon$, then
\begin{equation}
 \min_{\sigma\in\{-1,1\}}
 \|f-\sigma f_*\|_{L^1([0,1]^2)}
 \le4096\varepsilon^{1/72}.
 \label{eq:continuous-stability}
\end{equation}
\end{theorem}

We first prove two auxiliary estimates, keeping explicit constants for
the passage to finite grids.

For $h:[0,1]\to\{-1,1\}$ define
\[
 (\mathcal T h)(t)=\int_0^1\operatorname{sgn}(s-t)h(s)\,ds,
 \qquad D(h)=\frac13-\|\mathcal T h\|_2^2.
\]
Values of the sign kernel on its diagonal are immaterial.  A threshold
function is a sign function that is constant on each side of a point
$\xi\in[0,1]$, with opposite signs on those two intervals. Both orientations
and, through $\xi=0,1$, both constant functions are included.

\begin{lemma}[One-dimensional stability]
\label{lem:one-dimensional}
For every measurable sign function $h$, one has $D(h)\ge0$.  There is a
threshold function $b$ such that
\[
 \|h-b\|_1\le4D(h)^{1/3}.
\]
\end{lemma}

The bound $D(h)\ge0$ also follows from the sharp generalized Wirtinger
inequality of Croce and Dacorogna~\cite[Theorem~1.1]{Croce2003}.
The proof below gives an alternation identity from which we obtain the
quantitative threshold estimate.

\begin{proof}
Let
\[
 I(h)=\int_{0<r<s<t<1}
 \mathbf 1\{h(r)=h(t)\ne h(s)\}\,dr\,ds\,dt.
\]
The elementary identity
\[
 4\mathbf 1\{h(r)=h(t)\ne h(s)\}
 =1+h(r)h(t)-h(r)h(s)-h(s)h(t)
\]
gives
\begin{equation}
 D(h)=8I(h).
 \label{eq:alternation-deficit}
\end{equation}
To check the coefficient, after integrating out the third variable the
coefficient of $h(u)h(v)$, for $u<v$, in the triple expression is
$2(v-u)-1$.  On the other hand,
\[
 \int_0^1\operatorname{sgn}(u-t)\operatorname{sgn}(v-t)\,dt
 =1-2|u-v|.
\]
The constant triple integral is $1/6$, and the ordered two-variable
integral is half the corresponding symmetric integral.  These observations
give \eqref{eq:alternation-deficit}, including its constant term, by Fubini.

Let $e$ be the least measure of disagreement with a threshold.  For each
orientation, disagreement is continuous in its threshold position, and
$[0,1]$ is compact, so the minimum is attained.  We claim
$I(h)\ge e^3/64$.  There is nothing to prove if $e=0$.  Otherwise let $p$
be the measure of the positive set.  The constant negative threshold shows
that $p\ge e$.  By continuity of cumulative positive mass choose $a\le b$
so that the positive mass before $a$ and after $b$ is $e/4$ each.

If the negative mass between $a$ and $b$ is at least $e/4$, independently
choosing a positive point before $a$, a negative point between $a,b$, and
a positive point after $b$ gives $I(h)\ge(e/4)^3$.  Otherwise compare with
the threshold positive before $b$ and negative after $b$.  Its disagreement
is at least $e$, whereas the negative mass between $a,b$ plus the positive
mass after $b$ is less than $e/2$.  Thus the negative mass before $a$ is
at least $e/2$.  Comparing with the oppositely oriented threshold at $a$
similarly gives negative mass at least $e/2$ after $b$.  The positive mass
between $a,b$ is $p-e/2\ge e/2$.  Negative--positive--negative triples
therefore give the stronger bound $I(h)\ge e^3/8$.  The claim follows.
By \eqref{eq:alternation-deficit}, $e^3\le8D(h)$; since the sign-function
$L^1$ distance is $2e$, the asserted bound follows.
\end{proof}

Apply $\mathcal T$ in the two coordinates and write
\[
 F_u(u,v)=\int_0^1\operatorname{sgn}(s-u)f(s,v)\,ds,
 \qquad
 F_v(u,v)=\int_0^1\operatorname{sgn}(s-v)f(u,s)\,ds,
 \qquad M=\int_{[0,1]^2}f.
\]
Away from a null set, the indicator of comparability of $(u,v)$ and
$(u',v')$ is
$[1+\sgn(u'-u)\sgn(v'-v)]/2$.
The equality indicator is $[1+f(u,v)f(u',v')]/2$, and the ordered region
in~\eqref{eq:comparison-energy} is half the comparable region.
Expanding their product and using antisymmetry of the sign kernel gives
\[
 \mathcal E(f)=\frac{1+M^2-\langle F_u,F_v\rangle}{8}.
\]
In particular,
\begin{equation}
 8\bigl(\mathcal E(f)-1/12\bigr)
 =M^2+\frac12\|F_u-F_v\|_2^2
 +\frac12\left(\frac13-\|F_u\|_2^2\right)
 +\frac12\left(\frac13-\|F_v\|_2^2\right).
 \label{eq:comparison-deficit}
\end{equation}
Lemma~\ref{lem:one-dimensional} on each row and column makes every term on
the right nonnegative.  This proves the lower bound in
Theorem~\ref{thm:continuous}.  For $f_*$ the mean vanishes, both slice
families are thresholds, and direct integration gives $F_u=F_v$.  Hence its
energy is exactly $1/12$.

Under an excess bound $\varepsilon$, \eqref{eq:comparison-deficit} gives
\begin{equation}
 |M|\le\sqrt{8\varepsilon},\qquad
 \|F_u-F_v\|_2\le4\sqrt\varepsilon,
 \qquad
 \int_0^1 D(f(\cdot,v))\,dv\le16\varepsilon,
 \label{eq:deficit-consequences}
\end{equation}
and the analogous column estimate.  Jensen's inequality and
Lemma~\ref{lem:one-dimensional} show that the average row distance from a
threshold, and also the average column distance, is at most
$4(16\varepsilon)^{1/3}\le11\varepsilon^{1/3}$.  These distance functions are
measurable: one may take the infimum over rational threshold positions in
both orientations, using continuity in the threshold position.

\begin{lemma}[Approximation by a diagonal profile]
\label{lem:diagonal-approximation}
Suppose $0<\eta_{\rm sm}\le1$ and $\mathcal E(f)\le1/12+\eta_{\rm sm}^6$.  There is a measurable
$b:[0,2]\to\{-1,1\}$ such that, for $f_1(u,v)=b(u+v)$,
\[
 \|f-f_1\|_1\le128\eta_{\rm sm},
 \qquad \mathcal E(f_1)\le1/12+257\eta_{\rm sm}.
\]
\end{lemma}

\begin{proof}
It suffices first to treat $\eta_{\rm sm}<1/8$.  Extend $f$ by zero outside the square.
A threshold on $[0,1]$, extended by zero, has total variation at most four.
Approximating each row by a threshold and applying the triangle inequality
shows that translation by $r$ in the first coordinate changes the zero
extension in $L^1$ by at most
$22\varepsilon^{1/3}+4|r|$.  The same holds in the second coordinate.
The rowwise inequality can be integrated using the measurable
infimum described above.

Let $f_{\eta_{\rm sm}}$ be the convolution of the zero extension with the uniform
probability density on $[-\eta_{\rm sm},\eta_{\rm sm}]^2$.  Integrating the translation bound and
using $\varepsilon=\eta_{\rm sm}^6$ yields
\begin{equation}
 \|f-f_{\eta_{\rm sm}}\|_{L^1(\mathbb R^2)}
 \le44\eta_{\rm sm}^2+8\eta_{\rm sm}\le52\eta_{\rm sm}.
 \label{eq:box-error}
\end{equation}
On the interior square $Q_{\eta_{\rm sm}}=[\eta_{\rm sm},1-\eta_{\rm sm}]^2$ we also have
\begin{equation}
 \| (\partial_u-\partial_v)f_{\eta_{\rm sm}}\|_{L^1(Q_{\eta_{\rm sm}})}\le2\eta_{\rm sm}.
 \label{eq:directional-bound}
\end{equation}
Indeed, on the open unit square, in the sense of distributions,
\[
 (\partial_u-\partial_v)f
 =\tfrac12\partial_u\partial_v(F_u-F_v),
\]
because $\partial_u F_u=-2f$ and $\partial_v F_v=-2f$.  The mixed derivative
of the uniform averaging kernel is the signed measure at its four corners
with total variation $\eta_{\rm sm}^{-2}$.  Convolution on $Q_{\eta_{\rm sm}}$ therefore gives a bound
$\|F_u-F_v\|_1/(2\eta_{\rm sm}^2)\le2\eta_{\rm sm}$ by \eqref{eq:deficit-consequences}.
For almost every centre in $Q_{\eta_{\rm sm}}$, the averaging box is
contained in $(0,1)^2$, where the distributional identity above applies.
The box average of a bounded function is Lipschitz: translating either
coordinate changes the averaging rectangle by area at most a constant
times the translation distance.  Consequently the weak directional
derivative above agrees with its almost-everywhere derivative, and the
one-dimensional fundamental theorem applies on almost every diagonal slice.

For each $s$ average $f_{\eta_{\rm sm}}(u,s-u)$ over the interval of $u$ for which
$(u,s-u)\in Q_{\eta_{\rm sm}}$, and denote the average by $j(s)$.  Put $j(s)=0$ when
that interval is empty.  This defines a measurable function with
$|j|\le1$.  On an interval $I$ of length $\ell$, an absolutely continuous
function $g$ satisfies
\[
 \int_I|g-g_I|
 \le\frac1\ell\int_I\int_I|g(x)-g(y)|\,dx\,dy
 \le\ell\int_I|g'|.
\]
The diagonal coordinates $(s,u)$ have Jacobian one and their intervals
have length at most one.  Integrating this inequality and using
\eqref{eq:directional-bound} gives
$\|f_{\eta_{\rm sm}}-j(u+v)\|_{L^1(Q_{\eta_{\rm sm}})}\le2\eta_{\rm sm}$. Outside $Q_{\eta_{\rm sm}}$ the area is
at most $4\eta_{\rm sm}$ and the pointwise difference is at most two.  Together with
\eqref{eq:box-error}, this gives
\[
 \|f-j(u+v)\|_{L^1([0,1]^2)}\le52\eta_{\rm sm}+2\eta_{\rm sm}+8\eta_{\rm sm}=62\eta_{\rm sm}.
\]
Round $j$ to a sign, with a fixed convention at zero.  For a sign $z$ and
$r\in[-1,1]$, $|z-\operatorname{sgn}(r)|\le2|z-r|$.  Thus the resulting
$f_1$ has distance at most $124\eta_{\rm sm}\le128\eta_{\rm sm}$.

For any two sign functions $f,g$, expansion of the equality indicator
gives
\begin{equation}
 |\mathcal E(f)-\mathcal E(g)|
 \le\frac12\int_{[0,1]^2}|f-g|
       \bigl(uv+(1-u)(1-v)\bigr)\,du\,dv
 \le\|f-g\|_1.
 \label{eq:energy-lipschitz}
\end{equation}
Hence $\mathcal E(f_1)\le1/12+\eta_{\rm sm}^6+128\eta_{\rm sm}\le1/12+129\eta_{\rm sm}$, which is stronger
than the stated bound.  If $\eta_{\rm sm}\ge1/8$, any fixed diagonal sign function
suffices: the distance is at most two, the energy is at most $1/4$, and
both claimed bounds follow.
\end{proof}

\begin{proof}[Proof of Theorem~\ref{thm:continuous}]
Assume first that $\varepsilon>0$, put $\eta_{\rm sm}=\varepsilon^{1/6}$, and suppose
$\alpha=257\eta_{\rm sm}\le1/4096$.  Lemma~\ref{lem:diagonal-approximation} gives
$f_1(u,v)=b(u+v)$ with energy excess at most $\alpha$.
Let $t=\alpha^{1/4}\le1/8$.
By \eqref{eq:deficit-consequences} there is a row at height
$v\in[t,2t]$ whose one-dimensional deficit is at most $16\alpha/t$.
Its sign profile is $b$ restricted to $[v,v+1]$.  By
Lemma~\ref{lem:one-dimensional}, it has a threshold approximation of
error at most $4(16\alpha/t)^{1/3}<11\alpha^{1/4}$.
Restrict that threshold to $[v,1]$ and extend it to $[0,1]$ as a threshold.
The additional interval has length at most $2t$ and costs at most $4t$.
Thus $b$ on $[0,1]$ is within $15\alpha^{1/4}$ of a threshold $b_0$.
Selecting a row with $v\in[1-2t,1-t]$ in the same way gives a threshold
$b_1$ within $15\alpha^{1/4}$ of $s\mapsto b(1+s)$ on $[0,1]$.

Let $\widetilde b(s)=b_0(s)$ for $s<1$ and
$\widetilde b(s)=b_1(s-1)$ for $s\ge1$, and put
$f_2(u,v)=\widetilde b(u+v)$.  The density of $u+v$ on the square is
$s$ for $0\le s\le1$ and $2-s$ for $1\le s\le2$, so it is at most one.
Consequently
\begin{equation}
 \|f_1-f_2\|_1\le\zeta:=30\alpha^{1/4}.
 \label{eq:two-half-error}
\end{equation}
Each coordinate operator $\mathcal T$ has $L^2$ operator norm at most
one, by Cauchy--Schwarz on each unit interval and Fubini.  Moreover,
$\|f_1-f_2\|_2^2=2\|f_1-f_2\|_1$.  Writing
$\mathcal D=\mathcal T_u-\mathcal T_v$, the triangle inequality followed
by $(x+y)^2\le2x^2+2y^2$ gives
\begin{equation}
 \|\mathcal D f_2\|_2^2
 \le32\alpha+16\zeta\le512\alpha^{1/4}.
 \label{eq:two-half-operator}
\end{equation}

Set $m_0=|\{s\in[0,1]:b_0(s)=b_1(s)\}|$ and
$F(t)=\int_0^t(b_0(s)+b_1(s))\,ds$.  Directly integrating the two
coordinate operators yields
\begin{equation}
 \mathcal D f_2(u,v)=F(v)-F(u),
 \qquad
 \|\mathcal D f_2\|_2^2=2\int_0^1(F-\overline F)^2.
 \label{eq:agreement-variance}
\end{equation}
Here $\overline F=\int_0^1F(t)\,dt$ is the mean of $F$ on $[0,1]$.
Because $b_0,b_1$ are thresholds, the nonzero derivative of $F$ is
supported on at most two intervals, and is constantly $2$ or $-2$ on
each.  If $m_0>0$, one of those intervals has length
$\ell\ge m_0/2$.  For any constant $c$, the integral of $(F-c)^2$ on
that interval is at least
\[
 \min_d\int_0^\ell(2s-d)^2\,ds=\ell^3/3.
\]
Taking $c=\overline F$ in \eqref{eq:agreement-variance} and using
\eqref{eq:two-half-operator}, we obtain
$m_0^3\le6144\alpha^{1/4}$ and hence
\begin{equation}
 m_0\le19\alpha^{1/12}.
 \label{eq:agreement-bound}
\end{equation}

Replace $b_1$ by $-b_0$ and call the resulting diagonal function $f_3$.
Since the diagonal density is at most one,
$\|f_2-f_3\|_1\le2m_0$.  If $b_0$ is positive before $\xi$ and negative
after $\xi$, then
\[
 \int f_3=\int_0^1(2s-1)b_0(s)\,ds=-2\xi(1-\xi).
\]
Reversing the threshold orientation reverses the mean.  The distances
from $f_3$ to the two canonical phases are $2\xi$ and $2(1-\xi)$, in some
order, since the density weights of the two halves add to one.
Thus $\min\{2\xi,2(1-\xi)\}\le4\xi(1-\xi)=2|\int f_3|$.
With $d=\|f-f_2\|_1$, we have $|\int f_3|\le|\int f|+d+2m_0$, so
\[
 \min_\sigma\|f-\sigma f_*\|_1
 \le d+2m_0+2\left|\int f_3\right|
 \le3d+6m_0+2\left|\int f\right|.
\]
Using $d\le128\eta_{\rm sm}+30\alpha^{1/4}$,
\eqref{eq:agreement-bound}, and $|\int f|\le3\eta_{\rm sm}^3$, the right-hand side
is at most
\[
 384\eta_{\rm sm}+90\alpha^{1/4}+114\alpha^{1/12}+6\eta_{\rm sm}^3
 \le594\alpha^{1/12}.
\]
Here $0<\eta_{\rm sm}\le1$, $\alpha=257\eta_{\rm sm}\le1$, and each smaller power term is
bounded by the displayed coefficient times $\alpha^{1/12}$.
Since $257^{1/12}\le2$, this is at most
$1188\eta_{\rm sm}^{1/12}\le4096\varepsilon^{1/72}$.

If $257\eta_{\rm sm}>1/4096$, the distance between any two sign functions is at most
two.  The inequality $\eta_{\rm sm}^{1/12}\ge1/4$ at the cutoff, and monotonicity
thereafter, give the asserted estimate immediately.  Finally, if
$\varepsilon=0$, apply the positive-excess result for every positive
excess parameter.  The minimum of the distances to the two fixed phases
must then be zero. This proves the zero-excess case.
\end{proof}

\subsection{Finite rectangular grids with unary costs}

For integers $m,n\ge1$, let $J_{m,n}$ have vertex set
$\{0,\ldots,m-1\}\times\{0,\ldots,n-1\}$, where two vertices are adjacent
exactly when both coordinate orders strictly agree. Here $m,n$ denote
grid dimensions, independently of the source-graph size used elsewhere.
Write
\[
 V=mn,\qquad \rho=\frac1m+\frac1n,\qquad
 u_i=\frac{i+1/2}{m},\qquad v_j=\frac{j+1/2}{n}.
\]
Thus $V$ is the number of vertices of the comparison graph, not the
number of points of the eventual geometric drawing.  Define
\[
 c^+_{ij}=\mathbf1\{u_i+v_j<1\},\qquad c^-_{ij}=1-c^+_{ij},
 \qquad s^\sigma_{ij}=2c^\sigma_{ij}-1.
\]
In particular a diagonal cell centre has spin $-1$ in the positive phase.

Let $U_x:\{0,1\}\to\mathbb R$ be unary costs, and suppose
\[
 |U_x(1)-U_x(0)|\le\alpha V\quad\hbox{for all }x,
 \qquad \alpha\ge0.
\]
No bound is imposed on their additive constants.  Write
\[
 E(c)=\Phi_J(c)+\sum_xU_x(c_x),
\]
where $\Phi_J$ counts monochromatic edges of $J_{m,n}$.
For a set $F$ of grid vertices, $c\oplus F$ denotes the colouring obtained
by flipping exactly $F$.  Define the exact finite quantities
\begin{equation}
 d_x^\sigma=E(c^\sigma\oplus\{x\})-E(c^\sigma),\qquad
 L_\sigma=E(c^\sigma)+\sum_x\min(0,d_x^\sigma),\qquad
 L=\min(L_+,L_-).
 \label{eq:finite-linear-floor}
\end{equation}
For $w\ge0$ let
\[
 D_w=\{(i,j):|u_i+v_j-1|\le w\},\qquad
 C_w=\{(i,j):u_i+v_j\le w\ \hbox{or}\ u_i+v_j\ge2-w\}.
\]

\begin{theorem}[Finite-grid localization and energy bounds]
\label{thm:grid}
Let $a$ satisfy $\alpha+\rho\le a\le2^{-4096}$.
For every rectangular aspect ratio the following statements hold.
\begin{enumerate}
 \item Every minimizer $c$ of $E$ has a phase $\sigma$ such that, with
 $F=\{x:c_x\ne c_x^\sigma\}$,
 \[
 |F|\le52aV,\qquad F\subseteq D_{220a}\cup C_{220a}.
 \]
 \item The minimum energy belongs to
 \[
 [L-2^{32}a^3V^2,\ L+2^{32}a^3V^2].
 \]
 For either prescribed phase, the explicit set
 $F_\sigma^0=\{x:d_x^\sigma<0\}$ has size at most $52aV$ and
 \[
 E(c^\sigma\oplus F_\sigma^0)\le L_\sigma+2^{32}a^3V^2.
 \]
 \item For every $F\subseteq D_{220a}\cup C_{220a}$,
 \[
 E(c^\sigma\oplus F)\ge L_\sigma-2^{32}a^3V^2.
 \]
 This last conclusion requires no optimality of $c^\sigma\oplus F$.
 In fact it holds for arbitrary real unary costs, without their oscillation
 bound, whenever $\rho\le a\le1/65536$.
\end{enumerate}
For rational unary costs, both $L_\sigma$ and the two prescribed-phase
colourings are computable by polynomially many exact arithmetic operations.
\end{theorem}

\begin{proof}
\emph{From cells to the continuous objective.}
Extend a grid colouring to the sign function $f_c$ constant on its
$1/m$ by $1/n$ cells.  Comparable pairs in distinct rows and columns
contribute exactly $\Phi_J(c)/V^2$.  All other contributions are
nonnegative and their total uncoloured measure is
\[
 \frac14-\frac{\binom m2\binom n2}{m^2n^2}
 =\frac1{4m}+\frac1{4n}-\frac1{4mn}\le\frac\rho4.
\]
We will use the weaker bound
\begin{equation}
 0\le\mathcal E(f_c)-\Phi_J(c)/V^2\le\rho/2.
 \label{eq:cell-comparison}
\end{equation}
The functions $f_{c^+}$ and $f_*$ can disagree only where
$|u+v-1|\le\rho/2$.  This band has area at most $\rho$, whence
$\|f_{c^+}-f_*\|_1\le2\rho$, and the analogous assertion holds for
the negative phase.  By \eqref{eq:energy-lipschitz},
$\mathcal E(f_{c^+})\le1/12+2\rho$.
For a minimizer $c$, comparison with $c^+$ gives
$\Phi_J(c)\le\Phi_J(c^+)+\alpha V^2$.  Therefore
\[
 \mathcal E(f_c)\le1/12+\alpha+\tfrac52\rho\le1/12+5a.
\]
Theorem~\ref{thm:continuous} and the canonical cell discrepancy give a
phase whose disagreement fraction $\mu=|F|/V$ satisfies
\begin{equation}
 \mu\le2048(5a)^{1/72}+\rho\le2^{-20}.
 \label{eq:initial-fraction}
\end{equation}
For the final numerical bound, use $5<8$, $a\le2^{-4096}$, and
$\rho\le a$; the first term is less than $2^{11-4093/72}<2^{-45}$.

\emph{Canonical local fields.}
For $p=(u,v)$ let $Q_p$ be the union of the lower-left and upper-right
rectangles comparable to $p$.  With $s=1-u-v$, direct integration gives
\[
 \int_{Q_p}f_*=-s(1-|s|).
\]
For $s\ge0$, the lower-left rectangle contributes $uv$, and the
upper-right rectangle contributes $s^2-(1-u)(1-v)$; reflection gives the
other case.  The union of cells corresponding to the finite neighbours of
a cell centre differs from $Q_p$ only in its row and column strips,
of total area at most $\rho$.  The sign-function discrepancy contributes
at most another $2\rho$.  A single grid flip changes the cost by minus
the vertex spin times its neighbour-spin sum.  Consequently, in either
canonical phase,
\begin{equation}
 \left|\frac{d_x^{\rm grid}}V-\ell(u_i,v_j)\right|\le6\rho,
 \qquad
 \ell(u,v)=|1-u-v|(1-|1-u-v|).
 \label{eq:canonical-field}
\end{equation}
At a diagonal centre the continuous field is zero, so the strict tie
convention causes no change to this bound.  Including the unary cost gives
\begin{equation}
 d_x^\sigma\ge V\ell(u_i,v_j)-6aV.
 \label{eq:field-with-unaries}
\end{equation}
If $t\in[0,1]$ and $t(1-t)\le z$, then either $t\le2z$ or
$1-t\le2z$, by considering $t\le1/2$ and $t\ge1/2$ separately.
Thus a small value of $\ell$ places the vertex in a narrow diagonal strip
or in the two corner regions.

\emph{Finite strip and corner counts.}
A rectangle with side lengths $r,t$ contains at most
\[
 (mr+1)(nt+1)\le V(r+\rho)(t+\rho)
\]
cell centres.  For $\rho\le w\le1$, counting by rows gives
$|D_w|\le(2w+\rho)V\le3wV$, and bounding each corner triangle by
a square gives $|C_w|\le8w^2V$.  Two comparable points of $D_w$, ordered
coordinatewise, have the sum of their coordinate differences at most
$2w$.  For any strip point its strip neighbours therefore lie in two
squares of side $2w$, giving at most $18w^2V$ possibilities.  Hence for
every $F\subseteq D_w\cup C_w$,
\begin{equation}
 \deg_F(x)\le26w^2V\quad(x\in D_w),\qquad
 |E(J_{m,n}[F])|\le83w^3V^2.
 \label{eq:strip-counts}
\end{equation}
The degree bound adds the $18w^2V$ possible neighbours in $D_w$
and the at most $8w^2V$ vertices of $C_w$.
For the edge bound, strip--strip, strip--corner, and corner--corner
edges contribute at most $27w^3V^2$, $24w^3V^2$, and $32w^4V^2$,
respectively.  Counting an edge more than once when regions overlap only
increases these upper bounds.

\emph{Bootstrap from initial stability.}
For a disagreeing vertex $x\in F$, changing its colour back in the
minimizer has cost change
\[
 -d_x^\sigma-2\sum_{y\in F:\,xy\in E(J_{m,n})}s_x^\sigma s_y^\sigma
 \ge0.
\]
Thus $d_x^\sigma\le2\deg_F(x)$, and
\eqref{eq:field-with-unaries} yields
\begin{equation}
 \ell(u_i,v_j)\le6a+2\deg_F(x)/V\qquad(x\in F).
 \label{eq:local-optimality-profile}
\end{equation}
Initially $\deg_F(x)\le\mu V$.  The observation following
\eqref{eq:field-with-unaries} gives
\[
 F\subseteq D_w\cup C_w,
 \qquad w=4\mu+12a.
\]
By \eqref{eq:initial-fraction}, $\rho\le w<1/4$.  For
$x\in F\cap D_w$, \eqref{eq:strip-counts} improves
\eqref{eq:local-optimality-profile} to
$\ell\le6a+52w^2$.  In $D_w$ we have
$\ell\ge|1-u-v|/2$, so
$F\cap D_w\subseteq D_{12a+104w^2}$.  The direct row count and the
corner count now imply
\[
 \mu\le25a+216w^2
 \le25a+6912\mu^2+62208a^2.
\]
Since $\mu\le2^{-20}$ and $a\le2^{-4096}$, we have
$6912\mu^2\le\mu/2$ and $62208a^2\le a$.
Absorbing these terms proves $\mu\le52a$.
Substitution in the definition of $w$ gives $w\le220a$, establishing
the first conclusion.

\emph{Exact correction for simultaneous flips.}
For every flip set $F$ and arbitrary real unary costs,
\begin{equation}
 E(c^\sigma\oplus F)-E(c^\sigma)
 =\sum_{x\in F}d_x^\sigma
  +2\sum_{\{x,y\}\in E(J_{m,n}[F])}s_x^\sigma s_y^\sigma.
 \label{eq:simultaneous-flips}
\end{equation}
Indeed, an edge with both endpoints flipped keeps its equality status,
while its two single-flip contributions sum to
$-2s_x^\sigma s_y^\sigma$.  The second sum restores exactly this
discrepancy; unary terms already add independently.

If $F\subseteq D_{220a}\cup C_{220a}$ and
$\rho\le a\le1/65536$, \eqref{eq:strip-counts} bounds the absolute
correction by
\[
 2\cdot83\cdot220^3a^3V^2<2^{32}a^3V^2.
\]
The linear sum is at least $\sum_x\min(0,d_x^\sigma)$.
This proves the third conclusion for arbitrary unary costs. Applying it to the
localized minimizer proves the lower bound in the second conclusion.

For the upper bound in a prescribed phase, take exactly
$F_\sigma^0=\{x:d_x^\sigma<0\}$.  By
\eqref{eq:field-with-unaries}, this set is contained in
$D_{12a}\cup C_{12a}$, and its cardinality is at most
\[
 (36a+1152a^2)V\le52aV.
\]
Its linear sum in \eqref{eq:simultaneous-flips} is the defining sum for
$L_\sigma$, while its correction is bounded by
$2\cdot83\cdot12^3a^3V^2<2^{32}a^3V^2$.
Choosing the phase with smaller $L_\sigma$ proves the upper bound for
the minimum.  Finally, all quantities in \eqref{eq:finite-linear-floor}
are finite sums and single-vertex comparisons.  They require no search
over colourings and can be evaluated exactly in polynomially many
operations for rational unary costs.
\end{proof}

\subsection{Freezing probes with the grid colouring fixed}

The following proposition applies to any finite simple graph. In the
application its vertices are geometric edges and its adjacency relation
is proper crossing.  Thus its conclusion includes interactions between
the probes themselves.

\begin{proposition}[Probe freezing]
\label{prop:freezing}
Let a finite simple graph have vertex partition $X,T,Z$, with induced
subgraph on $X$ equal to $J_{m,n}$. Minimize the number of monochromatic
graph edges over all two-colourings of its vertices.
Use $V=mn$ and $\rho=1/m+1/n$.  Suppose every vertex of $X$ has at most
$\alpha V$ neighbours outside $X$, every vertex of $T$ has at most
$\beta V$ neighbours outside $X$, and
\[
 \alpha,\beta\ge0,\qquad \alpha+\rho\le a\le2^{-4096}.
\]
For $p\in T$ and $\sigma\in\{+,-\}$ let
\[
 h_p^\sigma=\sum_{x\in N(p)\cap X}s_x^\sigma.
\]
This is its colour-1 cost minus its colour-0 cost when the grid is
canonical.  Choose $y_p^+=1$ when $h_p^+<0$ and $y_p^+=0$ otherwise,
and put $y_p^-=1-y_p^+$.  These are best responses in both phases,
including their complemented zero-field tie convention.

Put $b=104a+\beta$, and suppose that $Q\ge0$ satisfies
\begin{equation}
 |\{p\in T:|h_p^\sigma|\le bV\}|
 \le Q(b+\rho)V
 \quad\hbox{for both phases}.
 \label{eq:soft-probe-count}
\end{equation}
Then every optimal colouring $\chi$ has a phase $\sigma$ and a colouring
$\psi$ such that only vertices of $T$ are changed, every probe is coloured
$y_p^\sigma$, and
\begin{equation}
 0\le\Phi(\psi)-\Phi(\chi)
 \le E_{\rm fr}:=2Q(b+\rho)bV^2.
 \label{eq:freezing-error}
\end{equation}
The unchanged grid differs from $c^\sigma$ on at most $52aV$ vertices,
all in $D_{220a}\cup C_{220a}$.  If $\OPT_{\rm fr}$ is the minimum over
colourings with the probes in either of the two frozen phases, then
\[
 0\le\OPT_{\rm fr}-\OPT\le E_{\rm fr}.
\]
It suffices for \eqref{eq:soft-probe-count} to have the stronger estimate
$|\{p:|h_p^\sigma|\le tV\}|\le Q(t+\rho)V$ for
$0\le t\le t_0$, together with $b\le t_0$.
\end{proposition}

\begin{proof}
Fix the colours outside $X$ in a global optimum.  Its restriction
to $X$ minimizes the resulting grid objective.  Each conditional unary
oscillation is at most the number of outside neighbours, hence at most
$\alpha V$.  Theorem~\ref{thm:grid} gives a phase $\sigma$ and a
disagreement set $F$ with $|F|\le52aV$ and the stated support.  Keep this
grid colouring fixed for the rest of the proof.

For every probe, its field from $X$ differs from $h_p^\sigma$ by at most
$2|F|\le104aV$.  Under any colouring of all other vertices, its field
from neighbours outside $X$ has absolute value at most $\beta V$.
Thus the full field differs from $h_p^\sigma$ by at most $bV$.
If $|h_p^\sigma|>bV$, the sign of the full field is the same as the
canonical sign, whatever the outside colours are.  Single-vertex
optimality of $\chi$ therefore forces this hard probe already to have
colour $y_p^\sigma$.

Recolour each remaining, soft probe to $y_p^\sigma$, in any order.
At every step the grid remains fixed.  The outside field still has
absolute value at most $\beta V$ even after earlier probe changes, since
all probe--probe neighbours were included in that degree bound.  The
cost increase for the current recolouring is at most the magnitude of
the current full field, which is bounded by
\[
 |h_p^\sigma|+104aV+\beta V\le2bV.
\]
There are at most $Q(b+\rho)V$ soft probes by
\eqref{eq:soft-probe-count}.  Summing the sequential cost changes gives
the upper bound in \eqref{eq:freezing-error}.  The lower bound follows
because $\chi$ is a global optimum and $\psi$ is a feasible colouring.
Only probes were changed, so the same grid and its localization remain.

The frozen feasible set is a subset of the original feasible set and
contains the constructed $\psi$.  Its optimum is therefore between
$\Phi(\chi)$ and $\Phi(\psi)$, proving the final optimum bound.
\end{proof}

\section{Geometric construction and rational coordinates}\label{sec:geometry}

We realize the originals, frames and programs with the crossing relations
outlined in Section~\ref{sec:overview}.

\subsection{Two convex chains}
For rational $D\geq10$, $x_A\geq1$, and $0<\kappa,\kappa_f\leq1/100$, set
\begin{equation}\label{eq:caps}
 P(y)=(-D-\kappa y^2,y),\qquad A(t)=(x_A+\kappa_f t^2,t),\qquad -1<y,t<1.
\end{equation}
Choose original points $P_j=P(y_j)$ in increasing parameter order, and frame points $A_i=A(t_i)$ in decreasing parameter order, with $0\leq j<o$ and $0\leq i<B$.

This is the standard double-chain configuration; see Hurtado, Noy and
Urrutia~\cite[Section~3, p.~341]{Hurtado1999}. We verify its crossing
relations for the coordinates above.

\begin{lemma}[The base drawing]\label{lem:double-chain}
The base points are in general position. The edges $A_iP_j$ induce the strict comparison grid:
\[
A_iP_j\sim A_kP_l\quad\Longleftrightarrow\quad (i<k\ \text{and}\ j<l)\ \text{or}\ (k<i\ \text{and}\ l<j).
\]
An edge within either chain does not cross an edge joining the chains.
\end{lemma}
\begin{proof}
For $a<b<c$, direct expansion gives
\begin{align*}
\orient(P(a),P(c),A(t))&=-(c-a)[D+x_A+\kappa_f t^2+\kappa((a+c)t-ac)]<0,\\
\orient(P(a),P(c),P(b))&=-\kappa(c-a)(b-a)(c-b)<0,\\
\orient(A(a),A(c),P(y))&=(c-a)[D+x_A+\kappa y^2+\kappa_f((a+c)y-ac)]>0,\\
\orient(A(a),A(c),A(b))&=\kappa_f(c-a)(b-a)(c-b)>0.
\end{align*}
The brackets are at least $D+x_A-3\kappa$ and $D+x_A-3\kappa_f$, respectively. These formulas cover every single-chain and mixed triple. With three points on one chain, the middle point lies strictly inside the triangle formed by the other two and any point on the other chain: the displayed chord test and the mixed tests for the two subpairs give its three side tests. This excludes an internal-chain edge crossing a joining edge. Two points on each chain form a convex quadrilateral, since the opposite chain lies strictly on one side of either same-chain chord. Its diagonals give the stated comparison relation, with the two parameter orders reversed.
\end{proof}

There are $S$ early source endpoint clusters and $S$ late clusters, each containing $W\geq2$ consecutive originals. Between them lie $D_m=512SW$ dummy originals, so $o=514SW$. Each cluster is an equally spaced interval of parameter diameter $\tau$ around a macro parameter. Early, dummy, and late parameters lie, respectively, in $(-3/4,-5/8)$, $(-1/2,1/2)$, and $(5/8,3/4)$. A matching pairs the early and late clusters; Section~\ref{sec:core} specifies it from the source permutation.

Every original chord has direction $(b-a)(-\kappa(a+b),1)$, while an original--frame vector has horizontal component at least $D+x_A$ and vertical component of absolute value at most two. These direction cones are separated. Programs will lie in $x<-D/2$, $|y|<2$. The crossing relations below hold for sufficiently small displacements and clone diameter $\tau$. We then choose both scales explicitly so that all relations hold simultaneously.

\subsection{Ordinary programs}
Take $8\mid B$ and $t_i=1-2(i+1)/(B+1)$. The artificial point $A_*=A(0)$ lies in the strict angular gap between the middle frame rows, as viewed from every original. At a macro point $p=P(a)$, an ordinary program has prototype
\begin{equation}\label{eq:ordinary-prototype}
r=p-\mu_o(A_*-p),\qquad \mu_o>0.
\end{equation}

\begin{lemma}[Crossings of ordinary program fans]\label{lem:ordinary-table}
For sufficiently small $\mu_o$, followed by sufficiently small $\tau$, every clone $p_s$ of the parent satisfies
\begin{equation}\label{eq:ordinary-own}
rA_k\sim p_sA_i\quad\Longleftrightarrow\quad
\begin{cases}i<k,&k<B/2,\\i>k,&k\geq B/2.\end{cases}
\end{equation}
Outside the parent block, the program fan has the neighbourhood of a virtual column inserted at that block. A core edge from a parent clone to a different macro or a dummy crosses exactly the first half of the program fans when its other endpoint is later, and exactly the last half when it is earlier. These tables hold on an open neighbourhood of the prototype.
\end{lemma}
\begin{proof}
Put $\epsilon_i=\sgn\orient(r,p_s,A_i)$. The strict central gap gives $\epsilon_i=1$ on the first half and $-1$ on the last half, uniformly for all clones after choosing $\tau$. Both endpoint tests for proper crossing are needed. When $i<k$, they require $\epsilon_i=\epsilon_k=1$; when $i>k$, they require $\epsilon_i=\epsilon_k=-1$. For $i=k$ there is a common endpoint. For $i\ne k$, the points $r$ and $p_s$ lie strictly on the same side of the line through $A_i$ and $A_k$: for $|y|\leq2$, this line satisfies
\[
x=x_A+\kappa_f((t_i+t_k)y-t_it_k)\geq x_A-5\kappa_f>0,
\]
whereas both points lie in the left strip. This proves~\eqref{eq:ordinary-own}. In every nonparent column the nonzero macro determinants persist under displacement and cloning, giving the virtual-column prefix or suffix.

For the core table, let $w=P(b)$ with $b\ne a$. Then
$\orient(p,w,r)=-\mu_o\orient(p,w,A_*)$, so $r$ and every frame point lie on opposite sides of $pw$. The other test compares the signs of $\orient(r,A_k,p)$ and $\orient(r,A_k,w)$. The former is negative on the first half and positive on the last; at zero displacement the latter has the sign of $b-a$, and this strict sign persists. This gives the claimed halves for the actual clones. All signs have positive margins, hence hold on an open neighbourhood. For short edges within the parent cluster, the fan-induced unary costs are evaluated directly in Section~\ref{sec:core}.
\end{proof}

Figure~\ref{fig:programs} compares the ordinary and rank constructions.
\begin{figure}[tbp]
\centering
\includegraphics{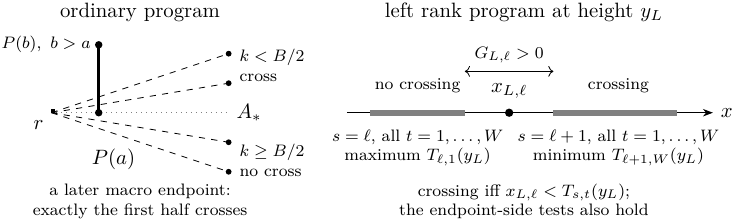}
\caption{Program fans acting on core edges. Left: the half-frame rule, with
$b>a$ in different macros. Right: $T_{s,t}(y_L)$ is the $x$-coordinate where
the chord $P(a_s)P(b_t)$ meets $y=y_L$; $s,t$ are the early/late clone ranks.
Each range covers all late ranks; the gap selects $s>\ell$.
The displacement and endpoint-side conditions are those of
Lemmas~\ref{lem:rank-table} and \ref{lem:common-scales}.}
\label{fig:programs}
\end{figure}

\subsection{Rank programs}
For a selected macro chord $P(a)P(b)$, $a<b$, write $d_0=b-a$ and let $a_s,b_t$, $1\leq s,t\leq W$, be its clone parameters. The selected-chord gap estimates use
\[
0<\mu_r<\frac1{100W},\qquad 0<\tau\leq\frac{\mu_rd_0}{100},\qquad
y_L=a+\mu_rd_0,\quad y_R=b-\mu_rd_0.
\]
The supporting horizontal coordinate of an actual selected chord at height $y$ is
\begin{equation}\label{eq:rank-support}
T_{s,t}(y)=-D-\kappa((a_s+b_t)y-a_sb_t).
\end{equation}

\begin{lemma}[Crossings of rank program fans]\label{lem:rank-table}
Fix all macro and dummy parameters, the selected pairing, the frame points, and $W\geq2$. There exists $\mu_0>0$, depending on these fixed data, such that for every $0<\mu_r<\mu_0$ there exists $\tau_0(\mu_r)>0$ for which the following statements hold whenever $0<\tau\leq\tau_0(\mu_r)$. We choose $\mu_0\leq1/(100W)$ small enough that $\mu_0(b-a)\leq|c-d|/4$ for every selected pair $a<b$ and every pair of distinct macro or dummy parameters $c,d$, and choose $100\tau_0(\mu_r)\leq\mu_r(b-a)$ for every selected pair.

At each internal gap $1\leq\ell<W$ there are open regions for left and right programs. Every left program fan crosses precisely the selected edges with $s>\ell$; every right fan crosses precisely those with $t\leq\ell$. The test is independent of the opposite clone rank.

The grid neighbourhood is the exact virtual column after the entire early parent block, or before the entire late parent block, independently of $\ell$. On an incident nonselected core edge to $P(w)$, every fan crosses if $a<w<b$, and none crosses otherwise.
\end{lemma}
\begin{proof}
At the unsteered centre $r_0=(1-\theta)P(a)+\theta P(b)$, where $\theta=\mu_r$ or $1-\mu_r$,
\[
\orient(r_0,A,P(a))=\theta\orient(P(a),P(b),A)<0,\quad
\orient(r_0,A,P(b))=(1-\theta)\orient(P(a),A,P(b))>0.
\]
Cloning and sufficiently small horizontal steering preserve these opposite endpoint signs for every frame point. The remaining crossing test is exactly
\begin{equation}\label{eq:rank-cross}
rA\sim P(a_s)P(b_t)\quad\Longleftrightarrow\quad x_r<T_{s,t}(y_r).
\end{equation}
Thus this is an internal segment crossing, not merely a support-line intersection.

At $y_L$, the support increases in $s$ and decreases in $t$. The full ranges over the opposite rank satisfy
\begin{align}
G_{L,\ell}
&=T_{\ell+1,W}(y_L)-T_{\ell,1}(y_L)\notag\\
&=\kappa[(a_{\ell+1}-a_\ell)(b_W-y_L)-(y_L-a_\ell)(b_W-b_1)]
\geq\frac{\kappa\tau d_0}{2W}>0.\label{eq:rank-gaps}
\end{align}
Indeed $a_{\ell+1}-a_\ell\geq\tau/W$, $b_W-y_L\geq0.98d_0$, and $y_L-a_\ell\leq1.005\mu_rd_0<0.011d_0/W$. Thus choose
\[
x_{L,\ell}=\tfrac12(T_{\ell,1}(y_L)+T_{\ell+1,W}(y_L)).
\]
At $y_R$, reflection and rank reversal give
$T_{1,\ell}(y_R)-T_{W,\ell+1}(y_R)\geq\kappa\tau d_0/(2W)$, so
\[
x_{R,\ell}=\tfrac12(T_{W,\ell+1}(y_R)+T_{1,\ell}(y_R))
\]
gives exactly $t\leq\ell$. These prototypes tend to the unsteered centres as $\tau\to0$, retaining the endpoint tests above.

For grid neighbours, marker-to-parent directions remain in the original-chord cone, outside the entire frame cone. The sign $\epsilon_i$ in the preceding proof is now constant over frame rows, positive for an original before the marker and negative for one after it. Consequently
\[
rA_k\sim P_jA_i\quad\Longleftrightarrow\quad
\begin{cases}i<k,&P_j\text{ precedes the marker},\\ i>k,&P_j\text{ follows it}.\end{cases}
\]
The left marker is above its whole parent block and the right marker below its whole parent block. First choose $\mu_0$ using the minimum of all macro gaps, so that the displacements $\mu_r(b-a)$ do not pass another macro and all nonparent macro ray signs persist. Once $\mu_r$ is fixed, choose $\tau_0(\mu_r)$ so that cloning and steering preserve these strict signs. The two selected-chord bounds alone do not control the nonparent rays.

For a nonselected incident chord, put $L_{a,w}(y)=-D-\kappa((a+w)y-aw)$. At the left unsteered centre,
\[
(r_0)_x-L_{a,w}((r_0)_y)=-\kappa\mu_r(b-a)(b-w).
\]
The support-side test excludes $w>b$. For $w<b$, the other endpoint test excludes $w<a$, since its sign at $P(w)$ tends to the sign of $w-a$, while the near-endpoint sign is negative. Thus precisely $a<w<b$ remains. The right support difference is $\kappa\mu_r(b-a)(a-w)$; the orientation sign at its parent endpoint $P(b)$ is positive and the other endpoint sign is negative precisely for $w<b$, giving the same interval. For each fixed $\mu_r$, all these support and ray signs are strict at $\tau=0$; the finite minimum of their allowed neighbourhoods gives a common positive upper bound on $\tau_0(\mu_r)$. Parent-cluster short edges and the selected partner are excluded from this last table. Lemma~\ref{lem:common-scales} gives a simultaneous rational choice, with the ray, partner, and curvature bounds in~\eqref{eq:ray-margins}--\eqref{eq:curvature-margins}.
\end{proof}

\begin{lemma}[Nonincident exclusion]\label{lem:nonincident}
A program fan does not cross a core edge whose endpoints both avoid its parent block.
\end{lemma}
\begin{proof}
For endpoints in different nonparent blocks, the fan at the parent macro is noncrossing by Lemma~\ref{lem:double-chain}, with nonzero macro determinants. Small displacement and then cloning preserve the exclusion. For a short edge in one nonparent block, the macro fan has positive distance from that block's centre. Sufficiently small displacement and clone neighbourhoods therefore keep the entire short edge disjoint from the fan.
\end{proof}

\subsection{Canonical response fields}
Let $N_g=Bo$, $\rho=1/B+1/o$, $u_k=(k+1/2)/B$, and $v_j=(j+1/2)/o$.
The canonical grid spins are the positive-phase spins of $J_{B,o}$,
namely $\sigma_{kj}=s^+_{kj}=2\1_{u_k+v_j<1}-1$.
The field of a fan is its grid cost in colour one minus its grid cost in colour zero. Its canonical response minimizes this cost, with colour zero at a tie. For a virtual column after $r_0$ originals, put $v=r_0/o$; its exact field is
\begin{align}
H_k&=\sum_{j<r_0}\sum_{i<k}\sigma_{ij}+\sum_{j\geq r_0}\sum_{i>k}\sigma_{ij},\notag\\
H_{k+1}-H_k&=\sum_{j<r_0}\sigma_{kj}-\sum_{j\geq r_0}\sigma_{k+1,j}.\label{eq:virtual-field}
\end{align}
There is no contribution from row $k$, whose fan shares a frame endpoint.

\begin{lemma}[Field and soft-row bounds]\label{lem:soft-fields}
Write $s_k=1-u_k-v$. Then
\begin{align}
|H_{k+1}-H_k-o(1-2|s_k|)|&\leq2+2o/B,\notag\\
|H_k/N_g+s_k(1-|s_k|)|&\leq6\rho.\label{eq:virtual-estimates}
\end{align}
A rank program has field $g_k=H_k$. For an ordinary program, choose a virtual position at a parent-block boundary. Then $|g_k-H_k|\leq2BW$, and away from the central ordinary switch the increment error is at most $2W$.

Suppose $B\geq o\geq1000$, $8\mid B$, and $W/o\leq1/514$. For either program type and $0\leq z\leq1/100$,
\begin{equation}\label{eq:soft-rows}
|\{k:|g_k|\leq zBo\}|\leq12Bz+3.
\end{equation}
An early program has response spin $+1$ on the middle $3B/4$ rows; a late program has spin $-1$ there. Each early half-frame sum is at least $B/4$ and its full sum at least $B/2$; the late sums have the opposite bounds. For $s$ programs the total soft count is at most $Q(z+\rho)N_g$, where $Q=12s/o$.
\end{lemma}
\begin{proof}
For $f_u(y)=2\1_{y<1-u}-1$, any midpoint sum over a prefix or suffix of column cells differs from $o$ times the integral by at most one: only the threshold cell matters, including a midpoint tie. Also changing $u$ by $1/B$ changes a suffix integral by at most $2/B$. Since
\[
\int_0^v f_u(y)\,dy-\int_v^1 f_u(y)\,dy=1-2|1-u-v|,
\]
the first estimate follows from~\eqref{eq:virtual-field}. The two rectangles $[0,u_k]\times[0,v]$ and $[u_k,1]\times[v,1]$ differ from the actual cell union only in row $k$, of area $1/B$. The cellwise canonical function differs from $2\1_{x+y<1}-1$ in $L^1$ by at most $2\rho$. Integrating the latter function over the rectangles gives $-s_k(1-|s_k|)$, proving even the bound $3\rho$. Only the $W$ parent columns differ for an ordinary program. Each prefix or suffix has absolute sum at most $B$ and changes by at most one away from its switch. This proves its two error bounds.

For early programs $v\leq1/514$, and for late programs $v\geq1-1/514$.
First apply the global field bound $|g_k-H_k|\leq2BW$, which holds
also at the ordinary switch. A soft row therefore satisfies
\[
|s_k|(1-|s_k|)\leq z+2W/o+6\rho<0.042<7/64.
\]
Thus every soft row belongs to one of the three intervals
$|s_k|\leq1/8$, $s_k\geq7/8$, or $s_k\leq-7/8$.
For early programs these intervals lie in $u\leq1/8$ or $u\geq3/4$;
for late programs they lie in $u\leq1/4$ or $u\geq7/8$.
They are therefore separated from the ordinary switch between rows
$B/2-1$ and $B/2$, at $u=1/2$.

Within each of these three intervals we may now apply the ordinary
increment error bound $2W$. The main increment $o(1-2|s_k|)$ has
constant sign and magnitude at least $3o/4$ there. The total increment
error is at most $2+2o/B+2W\leq6+2W$, leaving magnitude at least
$o/2$. Hence a field interval of length $2zBo$ contains at most
$4Bz+1$ rows in each of the three intervals, giving~\eqref{eq:soft-rows}.

On $u_k\in[1/8,7/8]$, an early rank has $1/9<s_k<8/9$, so $s_k(1-s_k)\geq8/81$. The error $2/514+12/o<0.016$ cannot change the negative field sign, hence the response colour is one. The late case reverses the sign. Each half-frame has $3B/8$ controlled and at most $B/8$ other rows, giving the claimed sums. Finally $12sBz+3s\leq(12s/o)(z+\rho)N_g$.
\end{proof}

The full rank weights lie in $[B/2,B]\subseteq[B/4,B]$. All internal
gaps at an endpoint have the same response sequence, preserved throughout
each open prototype region independently of other programs.

\subsection{All remaining edge classes}
Partition the complete graph into $X$ (original--frame), $T_f$ (program--frame), $C$ (original--original), $Y$ (program--original or program--program), and $I$ (frame--frame).
\begin{lemma}[Convex separation]\label{lem:separation}
The programs can all be placed strictly inside $\mathcal H=\{(x,y):\phi(x,y)\leq0\}$, where $\phi(x,y)=x+D+\kappa y^2$. Every open $X$ segment lies outside $\mathcal H$. Consequently $X$ has no crossing with $C$ or $Y$. The frame class $I$ has no crossing-graph neighbour outside itself.
\end{lemma}
\begin{proof}
The set $\mathcal H$ is convex. For $z_\lambda=(1-\lambda)P(a)+\lambda A(t)$,
\begin{equation}\label{eq:cap-ray}
\phi(z_\lambda)=\lambda[D+x_A+\kappa_ft^2+\kappa(2at-a^2)]+\kappa\lambda^2(t-a)^2.
\end{equation}
Its bracket is positive, so every open grid segment is outside. Substituting $\lambda=-\mu_o$ gives, for $0<\mu_o\leq1$,
$\phi(r)\leq-\mu_o(D+x_A-7\kappa)<0$ at an ordinary prototype. At a rank height, $a_s<y_r<b_t$ and
\[
\phi(T_{s,t}(y_r),y_r)=-\kappa(y_r-a_s)(b_t-y_r)<0.
\]
Its steered horizontal coordinate is a midpoint of two such supports; $\phi$ is affine in $x$, so the actual prototype is inside. If $\eta_{\rm cap}=\min_r(-\phi(r))>0$, perturbation by at most $\delta\leq\min(1/4,\eta_{\rm cap}/8)$ in both coordinates changes $\phi$ by less than $2\delta$ on $|y|\leq2$. Hence every final program remains inside. All $C,Y$ endpoints lie in $\mathcal H$, so convexity excludes their crossing with $X$.

For the frame component put $\psi(x,y)=x_A+\kappa_fy^2-x$ and use $\mathcal H_f=\{(x,y):\psi(x,y)\leq0\}$. It contains all $I$ segments. Every nonframe point $r$ has $x_r\leq-D$, $|y_r|\leq2$, and along a frame fan
\[
\psi((1-\lambda)A(t)+\lambda r)
=\lambda[x_A-x_r+\kappa_f(2ty_r-t^2)]+\kappa_f\lambda^2(y_r-t)^2>0
\]
for $0<\lambda\leq1$, since the bracket is at least $D+x_A-5\kappa_f$. Thus no fan crosses $I$. The remaining $C,Y$ segments lie in $x\leq-D$, disjoint from $\mathcal H_f\subseteq\{x\geq x_A\}$.
\end{proof}

The isolated component $I$ contributes $Z(B)$ to the optimum. Besides $I$--$I$, the only possible crossing types are
\[
X\!:\!X,\quad X\!:\!T_f,\quad T_f\!:\!T_f,\quad
C\!:\!C,\quad C\!:\!T_f,\quad Y\!:\!T_f,\quad Y\!:\!C,\quad Y\!:\!Y.
\]
Let $G_0$ be the crossing graph after deleting $Y$, and set
\begin{equation}\label{eq:nonfan-budget}
p=Bs,\qquad m_C=\binom o2,\qquad h=os+\binom s2,\qquad
R_Y=h(p+m_C)+\binom h2.
\end{equation}
Deletion removes only the last three crossing types, so
\begin{equation}\label{eq:nonfan-delete}
\OPT(G_0)\leq\OPT(\mathcal P)\leq\OPT(G_0)+R_Y.
\end{equation}
The upper bound extends any colouring of $G_0$ by arbitrary colours on $Y$, charging at most one per deleted constraint.

\subsection{A deterministic rational realization}\label{subsec:realization}
Fix $D=10$, $x_A=1$, and $\kappa=\kappa_f=1/100$. Having fixed the integer counts and source pairing, choose the early macros, physically ordered late macros, dummies, and frames by
\begin{equation}\label{eq:base-parameters}
a_j=-\tfrac34+\frac{j}{8(S+1)},\quad
b_l=\tfrac58+\frac{l}{8(S+1)},\quad
z_v=-\tfrac14+\frac{v}{2(D_m+1)},\quad
t_i=1-\frac{2(i+1)}{B+1}.
\end{equation}
Here $1\leq j,l\leq S$, $1\leq v\leq D_m$, and $0\leq i<B$. Let $\mathcal A$ be the set of early and late macro parameters and all dummy parameters; frame parameters are not included. Define
\begin{equation}\label{eq:explicit-scales}
\eta_{\rm sep}=\min\bigl(\{1\}\cup\{|a-c|:a,c\in\mathcal A,\ a\ne c\}\bigr),\quad
\mu=\frac{\eta_{\rm sep}^2}{10000(W+1)},\quad
\tau=\frac{\mu\eta_{\rm sep}^2}{1000(B+1)}.
\end{equation}
For $W\geq2$, use $\mu_o=\mu_r=\mu$ and clone parameters
$a_r=a+\tau\vartheta_r$, where $\vartheta_r=(2r-W-1)/(2(W-1))$.
Expand each early and late macro into its $W$ clones, retain each dummy
as one original, and enumerate all resulting parameters in increasing
order as $y_0,\ldots,y_{o-1}$.
There are $2SW$ candidate prototypes: $2S$ ordinary prototypes and $2S(W-1)$ internal rank prototypes, including those assigned zero copies later. The rank prototypes use the midpoint formulas of Lemma~\ref{lem:rank-table}.

\begin{lemma}[Simultaneous quantitative margins]\label{lem:common-scales}
The choices~\eqref{eq:explicit-scales} give all the crossing relations above for the actual clones and every candidate program. These scales are fixed by the base configuration before the program multiplicities are chosen.
\end{lemma}
\begin{proof}
The set $\mathcal A$ is finite with distinct parameters, so $0<\eta_{\rm sep}\leq1$. Direct substitution gives
\begin{gather}
0<\mu\leq1/10000,\quad 100W\mu<1,\quad
10000\mu\leq\eta_{\rm sep}^2,\quad 4800\mu<\eta_{\rm sep}^2,\notag\\
0<\tau\leq\eta_{\rm sep}/1000,\quad 1000\tau\leq\mu\eta_{\rm sep}^2,\quad
3(B+1)\tau/2\leq\mu,\quad
100\tau\leq\mu(b-a)
\label{eq:scale-margins}
\end{gather}
for any selected $a<b$, since $b-a\geq\eta_{\rm sep}$. Every clone remains in $(-1,1)$ and within $\tau/2$ of its macro. A nonparent original parameter $z$ remains on the same side of the parent macro $a$ as its own macro and obeys $|z-a|\geq\eta_{\rm sep}/2$. Consequently
\begin{equation}\label{eq:ray-margins}
5\mu\leq|z-a|,\qquad 2\mu(b-a)\leq|z-a|.
\end{equation}
If $z$ also avoids the selected partner macro $b$, then $|b-z|\geq\eta_{\rm sep}/2$ and
\begin{equation}\label{eq:partner-margin}
20\tau\leq\mu(b-a)|b-z|.
\end{equation}
For two nonparent actual parameters $z,v$, even when they are clones of the same macro, their distances to $a$ are each at least $\eta_{\rm sep}/2$. Thus
\begin{equation}\label{eq:curvature-margins}
1200\mu<|z-a||v-a|,\qquad
200\mu(b-a)<|z-a||v-a|.
\end{equation}
The second inequality uses $b-a\leq2$.

We check both endpoint-side tests for each crossing. In the ordinary own-column test, the unperturbed parent determinant has magnitude at least $10\mu|t_i|\geq10\mu/(B+1)$. Moving its original parameter from $a$ to a clone $c$ changes it by at most $25|c-a|$. The bound $3(B+1)|c-a|\leq\mu$ therefore preserves its sign. To obtain the constant 25, expand the determinant difference:
\[
\orient(r,P(c),A(t))-\orient(r,P(a),A(t))
=(c-a)\bigl(r_x-A(t)_x-(a+c)(t-r_y)/100\bigr);
\]
here $|r_x|\leq22$, $|r_y|\leq2$. For a nonparent ray, ordinary displacement changes the determinant by at most $36\mu$, while its macro magnitude is at least $10|z-a|$. Inequality~\eqref{eq:ray-margins} preserves that sign. These bounds establish the full ordinary grid table, including its central switch.

For a core support $L_{c,z}(y)=-10-((c+z)y-cz)/100$, the ordinary offset is exactly
\[
r_x-L_{c,z}(r_y)
=-\frac{(a-c)(a-z)}{100}-\mu C(a,c,z),
\qquad C(a,c,z)=11+\frac{a^2-ac-az}{100},
\]
with $10\leq C(a,c,z)\leq12$ for parameters in $[-1,1]$.
When $c$ is a parent clone, the first term has absolute value at most $\tau/100$; the second is strictly negative and dominates it. Combining with the parent and nonparent ray signs gives the ordinary core half-frame table. For two nonparent endpoints, the product margin~\eqref{eq:curvature-margins} preserves the macro support sign when the parent lies between them; if they are on one side, both ray signs agree. In either case the fan is disjoint from the core edge, including a short edge in a different cluster.

For rank programs, the last inequality of~\eqref{eq:scale-margins} and $100W\mu<1$ are exactly the full-range gap hypotheses in~\eqref{eq:rank-gaps}. The midpoint has horizontal steering from its unsteered macro-chord support at most $\tau/50$, obtained by expanding~\eqref{eq:rank-support} with endpoint deviations at most $\tau/2$. Its $L^\infty$ displacement from its parent is at most $\mu(b-a)$. For either actual selected endpoint $c$, the height gap is at least $\tau$ and
\[
|x_r-P(c)_x|\leq |y_r-c|/10.
\]
Indeed for a support through parameters $u,v$,
\[
L_{u,v}(y)-P(c)_x
=\bigl((c-u)(c-v)-(u+v)(y-c)\bigr)/100;
\]
one endpoint differs from $c$ by at most $\tau$, $|c-u|,|c-v|,|u+v|\leq2$, and averaging preserves the bound. The horizontal distance to every frame point is at least 11, while its vertical distance is at most two. The endpoint determinant thus has the strict sign prescribed by whether $c<y_r$ or $c>y_r$. This proves the endpoint-side test in~\eqref{eq:rank-cross} uniformly over both ranks.

A rank displacement bounded by $\delta_r=\mu(b-a)$ changes a nonparent ray determinant by at most $14\delta_r$, compared with macro magnitude at least $10|z-a|$; use~\eqref{eq:ray-margins}. For a nonselected incident edge from parent clone $c$ to actual $z$, the left support offset has the explicit error bound
\[
\left|r_x-L_{c,z}(r_y)+\frac{\mu(b-a)(b-z)}{100}\right|
\leq\frac{3\tau}{100}.
\]
It is the steering error $\tau/50$ plus clone error at most $\tau/100$. Inequality~\eqref{eq:partner-margin} preserves its sign, so the ray and support tests together give exactly $a<z<b$. The right statement follows by reflection. Finally any rank displacement of size $\delta_r$ changes a nonincident chord offset by at most $2\delta_r$, whereas its initial curvature offset is $|(a-z)(a-v)|/100$. Inequality~\eqref{eq:curvature-margins} again preserves the sign. Endpoints on the same side give equal ray signs, and endpoints straddling the parent put program and frame on the same side of their chord. This proves rank nonincidence in all cases.

The strict inwardness formulas in Lemma~\ref{lem:separation} apply directly to these ordinary and steered rank prototypes. Their heights have absolute value below two. Thus all prototypes are valid simultaneously. The argument uses only fixed base data and the candidate gap index.
\end{proof}

Table~\ref{tab:geometric-certificates} collects the crossing tests used in the construction. For a rank program, $d_0=b-a$; all quantities in the last column use the common scales~\eqref{eq:explicit-scales}. Short edges in nonparent clusters are included in the nonincident row.
\begin{table}[htbp]
\centering
\small
\setlength{\tabcolsep}{4pt}
\begin{tabular}{@{}p{0.25\linewidth}p{0.36\linewidth}p{0.33\linewidth}@{}}
\toprule
Required relation & Strict certificate & Controlling bound \\
\midrule
Grid $X$--$X$ & Opposite-chain chord signs and diagonal order & Lemma~\ref{lem:double-chain} \\
Ordinary fan--grid & Parent ray signs switch only at the middle rows; other columns keep their macro ray signs & $3(B+1)\tau/2\leq\mu$ and~\eqref{eq:ray-margins} \\
Rank fan--grid & Parent rays put the marker beyond the whole block; nonparent rays retain their order & $100\tau\leq\mu d_0$ and~\eqref{eq:ray-margins} \\
Rank fan--selected core & Opposite endpoint rays and the midpoint between full support ranges & Support gap $\geq\kappa\tau d_0/(2W)$ in~\eqref{eq:rank-gaps} \\
Ordinary fan--incident core; rank fan--nonselected incident core & Parent and other endpoint ray signs, together with the chord support offset & Ordinary offset dominated by $\mu C$; rank error $\leq3\tau/100$, controlled by~\eqref{eq:partner-margin} \\
Fan--nonincident core & Equal ray signs, or program and frame on the same side of the core chord & Curvature bounds~\eqref{eq:curvature-margins} \\
Fan--parent short edge & Both endpoints have the same nonzero fan ray sign & Ordinary parent bound; rank endpoint height gaps $\geq\tau$ \\
$X$--$(C\cup Y)$ exclusion and isolation of $I$ & Strict cap separation for open fan segments; convexity for internal segments & Lemma~\ref{lem:separation} and the inward margin $\eta_{\rm cap}$ \\
\bottomrule
\end{tabular}
\caption{Crossing relations, strict certificates, and the bounds that preserve them.}
\label{tab:geometric-certificates}
\end{table}

Compute each candidate's exact grid neighbours and response sequence. These determine the ordinary and rank weights. The greedy and ceiling rules of Section~\ref{sec:core} then determine the multiplicities and hence $s$, with the base coordinates and scales fixed.

Let $\mathcal B$ be the set of all fixed original and frame points, and $\mathcal R$ the finite set of used prototypes. Define positive rational margins
\begin{align*}
\eta_{\rm or}&=\min\bigl(\{1\}\cup
 \{|\orient(u,v,r)|:u,v\in\mathcal B,\ r\in\mathcal R,\ \orient(u,v,r)\ne0\}\bigr),\\
\eta_{\rm cap}&=\min\bigl(\{1\}\cup\{-\phi(r):r\in\mathcal R\}\bigr),\\
\eta_{\rm ht}&=\min\bigl(\{1\}\cup\{2-|r_y|:r\in\mathcal R\}\bigr),\\
L_{\mathcal B}&=1+2\sum_{u\in\mathcal B}(|u_x|+|u_y|),\\
\delta&=\min\{1/4,\eta_{\rm cap}/8,\eta_{\rm ht},\eta_{\rm or}/(8L_{\mathcal B})\}.
\end{align*}
Every nonzero one-program determinant retains its sign under a coordinate change at most $\delta$: its change is bounded by $\delta(|u_x-v_x|+|u_y-v_y|)<L_{\mathcal B}\delta\leq\eta_{\rm or}/8$. At the insertion points below the displacement is less than $\delta/2$ per coordinate. Hence the cap and height bounds are preserved as well.

An unused zero determinant need not retain its value. Every prescribed crossing has four strict determinant tests. A pair sharing a base endpoint remains noncrossing by definition; each other prescribed exclusion has either a strict same-side support pair or strict same-side fan rays. The quantitative proof above supplies these certificates, including parent-cluster short edges, whose two fan rays have the same nonzero sign. Thus preserving the nonzero determinants suffices.
Let $N=B+o+s$ be the final total point count.
Set $m=N^2+N+2$. For each required copy of prototype $r$, test in lexicographic order the $m^2$ rational candidates
\begin{equation}\label{eq:fixed-grid}
r+\delta\left(\frac{i}{m+1}-\frac12,\frac{j}{m+1}-\frac12\right),\qquad 1\leq i,j\leq m.
\end{equation}
Reject coincidences with existing points and collinearity with any existing pair. At most $\binom N2$ forbidden lines each meet this grid in at most $m$ points, and at most $N$ further points are forbidden. Since $\binom{N}{2}m+N<m^2$, some candidate survives. Induction gives distinct points in general position. The same radius and denominator are used throughout. All required one-program tables and response colours persist; actual fan--fan crossings and short-edge unary costs are evaluated after insertion, while program nonfans are covered by~\eqref{eq:nonfan-budget}.

\begin{proposition}[Effective realization]\label{prop:realization}
The construction is a deterministic polynomial-time algorithm with rational coordinates of polynomial binary length. It preserves all the stated crossing relations simultaneously and yields general position. A common positive scaling yields integer coordinates of polynomial binary length with the same crossing relation.
\end{proposition}
\begin{proof}
Proposition~\ref{prop:program-count} bounds the actual program count,
and there are $2SW$ candidate prototypes. Both quantities, as well as
$B$ and $o$, are bounded by fixed polynomials in the source size $n$.
The base coordinates have
$O(\log(B+1)+\log(W+1)+\log(S+1)+\log(D_m+1)+1)$ bits.
Since the formulas for $\eta_{\rm sep},\mu,\tau$ use a bounded
number of rational operations, and the internal rank gaps use explicit
midpoint formulas, the prototype coordinates and determinant margins
also have polynomial bit length. Taking a finite minimum preserves
this bound. The sum defining $L_{\mathcal B}$ has polynomially many
terms of polynomial bit length; a common denominator has bit length
at most the sum of their denominator bit lengths. Consequently both
$L_{\mathcal B}$ and $\delta$ have polynomial bit length.

The integer greedy recurrence is evaluated with cached cumulative weights. Inserting at most $N$ programs tests $O(N^4)$ candidates per insertion, each against $O(N^2)$ exact equations and determinants, for $O(N^7)$ rational operations. All intermediate bit lengths are polynomial; standard integer arithmetic therefore implements the whole algorithm in polynomial bit time. Earlier inserted points only reject candidates and never redefine later scales. Finally the product of the positive coordinate denominators has bit length at most their sum. Scaling by this product preserves strict determinant signs and gives integer coordinates with polynomial bit length.
\end{proof}

\section{Encoding the source graph in the core}\label{sec:core}

We minimize the original--original edge cost with the program fans $T_f$
fixed to their canonical grid responses $t^+$. The palette colours follow
from this minimization.

\subsection{Source matching and amplification}
Let $G$ and $g(G)$ be as in the preliminaries. Place its two endpoint orders in consecutive blocks of the original cap, reversing one order. The resulting chords alternate exactly for adjacent source vertices. Replace each source chord by $q$ macro copies, ordered forward at the early endpoint and backward at the late endpoint. Copies of a vertex are nested; two distinct source classes have either all $q^2$ crossings or none. Write $S=nq$ and label the selected macro chords $p_jq_j$, $1\leq j\leq S$, by early endpoint order.

Each macro endpoint is replaced by $W$ consecutive originals, denoted $p_{j,r},q_{j,t}$, $1\leq r,t\leq W$.
Both clone ranks increase with the parameter; the early rank $r$ is the
index $s$ in Lemma~\ref{lem:rank-table}.
There are $D_m=512SW$ dummy points between the endpoint blocks. All original points are in convex position on the same cap, so their chords cross precisely when their endpoints alternate.

Partition the core $C$ into selected edges
\[
\mathcal S=\bigcup_{j=1}^S\mathcal S_j,\qquad
\mathcal S_j=\{p_{j,r}q_{j,t}:1\leq r,t\leq W\},
\]
short edges $\mathcal U$ within one endpoint cluster, dummy--dummy edges $\mathcal D$, and the remaining palette edges $C_{\rm pal}$. Define the proposed palette by
\begin{equation}\label{eq:palette}
\chi_{\rm pal}(e)=
\begin{cases}
0,&e\text{ joins two early clusters or an early cluster to a dummy},\\
1,&e\text{ joins two late clusters or a dummy to a late cluster},\\
\1_{j<k},&e=p_{j,r}q_{k,t},\ j\ne k.
\end{cases}
\end{equation}
Every palette edge falls in exactly one case.

\begin{lemma}[Independent source amplification]\label{lem:source-amplification}
After deleting the crossing constraints internal to each $\mathcal S_j$, the graph induced on $\mathcal S$ is the independent $qW^2$-fold blowup of $G$. Its minimum monochromatic count is $\Delta g(G)$, where $\Delta=q^2W^4$.
\end{lemma}
\begin{proof}
Consecutive endpoint clusters preserve all crossings between different macro chords. The $q$ macro copies of one vertex remain pairwise noncrossing, and the deletion makes each individual bundle independent. Thus each source vertex becomes an independent class of $qW^2$ variables with identical neighbours; source adjacency gives a complete join.

The normalized MaxCut value is invariant under independent blowups
\cite[p.~268]{Borgs2006}. Thus an $m$-fold blowup multiplies both the edge
count and MaxCut by $m^2$, and hence also multiplies $g(G)$ by $m^2$.
Taking $m=qW^2$ gives the claimed value.
\end{proof}

\subsection{Unary costs from frozen fans}
Every unary coefficient means cost in colour one minus cost in colour zero; a frozen neighbour of colour $a$ contributes $2a-1$. Define
\begin{align}
E_C(x;t^+)
&=\sum_{\substack{\{e,f\}\subseteq C\\e\sim f}}\1_{x_e=x_f}
+\sum_{\substack{e\in C,\ f\in T_f\\e\sim f}}\1_{x_e=t_f^+},
&
M_C&=\min_{x:C\to\{0,1\}}E_C(x;t^+).
\label{eq:core-energy}
\end{align}
Lemmas~\ref{lem:ordinary-table}--\ref{lem:soft-fields} give ordinary
half-frame weights in $[B/4,B]$ on incident nonshort edges, positive at early
endpoints and negative at late endpoints. Rank programs add $w_\ell\in[B/4,B]$ on the early
suffix $r>\ell$ or $-w_\ell$ on the late prefix $t\leq\ell$.
Nonincident edges receive no contribution; short-edge unaries are
computed directly. Each response weight depends only on its prototype
and the base grid, independently of other multiplicities and core colours.

\begin{lemma}[Affine palette field]\label{lem:affine-palette}
The palette contribution to a selected edge is
\begin{equation}\label{eq:affine-palette}
\sum_{\substack{f\in C_{\rm pal}\\f\sim p_{j,r}q_{j,t}}}(2\chi_{\rm pal}(f)-1)
=\beta_j+\gamma_{j,L}(r-1)+\gamma_{j,R}(t-1),
\end{equation}
where the integer coefficients are computable by crossing counts and obey
$|\beta_j|\leq o^2$ and $-o\leq\gamma_{j,L},\gamma_{j,R}<0$.
\end{lemma}
\begin{proof}
A palette edge avoiding both parent clusters contributes independently of $r,t$. One with an endpoint in the early parent crosses according to whether that endpoint precedes or follows $r$. Its other endpoint avoids the late parent, and its colour is constant over the early cluster; its sum is affine in $r$ alone. The late cluster similarly gives a function of $t$. An edge between the two parent clusters is selected, and one inside a cluster is short, so neither contributes a mixed term to the palette field.

The dummy neighbours give exactly
\[
-D_m(r-1)+D_m(W-t).
\]
The early--dummy edges have colour zero and cross when their parent rank is below $r$; the dummy--late edges have colour one and cross when their parent rank is above $t$. Changing either rank changes at most one contribution for each nondummy endpoint, of which there are at most $2SW$. Thus either slope lies in
$[-D_m-2SW,-D_m+2SW]\subseteq[-o,0)$. Finally the number of possible palette neighbours is less than $o^2$, proving the intercept bound.
\end{proof}

\begin{lemma}[Greedy rank correction]\label{lem:rank-correction}
Each slope in Lemma~\ref{lem:affine-palette} can be corrected with at most $4oW/B+4$ rank programs at that endpoint. After correction the field on bundle $j$ is
\begin{equation}\label{eq:corrected-field}
\beta_j^{\rm corr}+\eps_{j,L}(r)+\eps_{j,R}(t),\qquad
0\leq\eps_{j,L}(r),\eps_{j,R}(t)<B,\qquad |\beta_j^{\rm corr}|\leq4o^2.
\end{equation}
\end{lemma}
\begin{proof}
For $\gamma<0$, let $f(r)=-\gamma(r-1)$, so $f(1)=0$ and $f(W)\leq oW$. Starting with $R(1)=0$, at gap $r-1$ add the least nonnegative number of its weight-$w_{r-1}$ programs that makes the new cumulative weight $R(r)$ at least $f(r)$. Inductively
\begin{equation}\label{eq:greedy}
f(r)\leq R(r)<f(r)+B.
\end{equation}
If copies are added, the final overshoot is smaller than the last weight, at most $B$. If none is added, the preceding overshoot remains below $B$ because the target increased. Each copy weighs at least $B/4$, while $R(W)<oW+B$, giving the stated count.

At the early endpoint the coefficient is $R(r)$, which cancels the slope with error in $[0,B)$. At the late endpoint the negative prefixes give $-R(W)+R(t)$. Their initial offset must be retained:
$\beta_j^{\rm corr}=\beta_j-R_{j,R}(W)$. Since $W\leq o$ and $B\leq o^2/1024$,
$|\beta_j^{\rm corr}|\leq o^2+oW+B\leq4o^2$. Only the $W-1$ internal gaps are used.
\end{proof}

\subsection{Forcing the palette inside the core}
Use the scales in~\eqref{eq:parameters} and put $\theta_j=(j-1)T$ for $1\le j\le S$.
For early--late edges let $u_{L,j}$ be the ordinary positive half-frame weight at early endpoint $j$, and $u_{R,j}$ the magnitude of the negative weight at its late partner. Choose
\begin{equation}
C_{L,j}=\ceil[\Bigg]{\frac{\Lambda_{\rm pal}+\theta_j-\beta_j^{\rm corr}}{u_{L,j}}},
\qquad
C_{R,j}=\ceil[\Bigg]{\frac{\Lambda_{\rm pal}+\theta_j}{u_{R,j}}}.
\label{eq:ordinary-counts}
\end{equation}
Both targets are positive. The rounding errors $e_{L,j},e_{R,j}$ in
Table~\ref{tab:selected-field} lie in $[0,B)$; the cumulative weights
$R_{j,L},R_{j,R}$ are those of Lemma~\ref{lem:rank-correction}.

\begin{table}[ht]
\centering
\small
\begin{tabular}{@{}ll@{}}
\toprule
Contribution & Unary coefficient on $p_{j,r}q_{j,t}$\\
\midrule
Palette & $\beta_j+\gamma_{j,L}(r-1)+\gamma_{j,R}(t-1)$\\
Rank programs & $R_{j,L}(r)-R_{j,R}(W)+R_{j,R}(t)$\\
Early ordinary programs & $C_{L,j}u_{L,j}=\Lambda_{\rm pal}+\theta_j-\beta_j^{\rm corr}+e_{L,j}$\\
Late ordinary programs & $-C_{R,j}u_{R,j}=-\Lambda_{\rm pal}-\theta_j-e_{R,j}$\\
\midrule
Sum & $\eps_{j,L}(r)+\eps_{j,R}(t)+e_{L,j}-e_{R,j}$\\
\bottomrule
\end{tabular}
\caption{Cancellation on a selected edge, with the palette fixed to
$\chi_{\rm pal}$. The late rank offset is included in
$\beta_j^{\rm corr}=\beta_j-R_{j,R}(W)$.}
\label{tab:selected-field}
\end{table}

\begin{lemma}[Palette forcing and selected residuals]\label{lem:palette}
Every minimizer defining $M_C$ agrees with $\chi_{\rm pal}$ on $C_{\rm pal}$. Once the palette and fans are fixed, the fixed palette and fan neighbours of each selected edge give a unary coefficient $r_e$ with $|r_e|<4B$; in particular $|r_e|\leq10B$, the bound used in Proposition~\ref{prop:core-encoding}.
\end{lemma}
\begin{proof}
For a nonselected early-$j$--late-$k$ edge, $j\ne k$, the ordinary coefficient is
\[
\theta_j-\theta_k-\beta_j^{\rm corr}+e_{L,j}-e_{R,k}.
\]
The first term has absolute value at least $T$. All adverse terms, including arbitrary remaining core colours, have magnitude at most
\begin{align*}
|\beta_j^{\rm corr}|+2B+2(4oW/B+4)B+\binom{o-2}{2}
&\leq\tfrac92o^2+8oW+10B\\
&<5o^2<T.
\end{align*}
The last strict bound uses $o\geq514W$ and $B\leq o^2/1024$. The rank term counts both incident endpoint families; nonincident programs contribute nothing. The binomial term bounds the full core crossing degree, without assuming any other palette colour. A wrong colour can therefore be improved by a single flip. The preferred colour is zero for $j>k$ and one for $j<k$.

For an edge on one side, the two ordinary fields have the same sign. Each ordinary multiplicity is at least $\Lambda_{\rm pal}/(2B)$ because its target is at least $\Lambda_{\rm pal}-4o^2$. Either half-frame then contributes magnitude at least $\Lambda_{\rm pal}/8$ from each parent. A side--dummy edge has one such contribution. All adverse rank and core terms have magnitude at most
\[
8oW+8B+\binom{o-2}{2}<o^2<\Lambda_{\rm pal}/8.
\]
Thus early-side edges prefer zero and late-side edges prefer one. This covers the palette simultaneously.

For the selected edge $p_{j,r}q_{j,t}$, the ordinary counts cancel $\beta_j^{\rm corr}$ in~\eqref{eq:corrected-field}, leaving exactly
\begin{equation}\label{eq:selected-residual}
r_e=\eps_{j,L}(r)+\eps_{j,R}(t)+e_{L,j}-e_{R,j}.
\end{equation}
All four errors have magnitude below $B$, proving the residual bound.
\end{proof}

Palette forcing applies to~\eqref{eq:core-energy}. The full-drawing lower
bound uses only that the post-freezing core energy is at least $M_C$.

\subsection{Dummy edges and local constraints}
\begin{lemma}[Exact dummy elimination]\label{lem:dummy}
With the palette fixed, the optimized dummy-edge contribution, including its palette interactions, is
\begin{equation}\label{eq:dummy}
Z(D_m)+SW\binom{D_m}{3},
\end{equation}
independently of all selected and short-edge colours.
\end{lemma}
\begin{proof}
A dummy chord with ranks $a<b$ crosses no early--late selected chord and no endpoint-cluster short edge. Its palette neighbours are exactly the side--dummy edges whose dummy endpoint is strictly between $a,b$. It has $SW(b-a-1)$ such neighbours of each colour. Its palette cost is consequently the constant $SW(b-a-1)$ in either colour. There are no program-fan neighbours by Lemma~\ref{lem:nonincident}. Summing over $a<b$ gives $SW\binom{D_m}{3}$. What remains is the full convex drawing on $D_m$ points, with optimum $Z(D_m)$ by the two-page complete-graph theorem \cite{Abrego2013}. Its endpoint-alternation predicate agrees with that of convex straight-line chords.
\end{proof}

Let $Q_{\rm loc}$ be the crossing constraints internal to one selected bundle, between selected and short edges, and between two short edges. These are crossing-graph constraints, not removed geometric edges.

\begin{lemma}[Conservative local budget]\label{lem:local-count}
\begin{equation}\label{eq:local-budget}
|Q_{\rm loc}|\leq2SW^4.
\end{equation}
After deleting these constraints and fixing the palette, every short edge has only unary costs. The dummy subsystem remains independent, and the remaining selected-edge graph is that of Lemma~\ref{lem:source-amplification}.
\end{lemma}
\begin{proof}
In one selected bundle, an internal crossing is specified by two early and two late ranks, giving at most $W^4/4$ choices. For a source--short crossing, the short chord belongs to one of that source's two endpoint clusters and its endpoints surround the source endpoint. An upper bound, even without this surrounding condition, is
$2\binom W2 W^2\leq W^4$ per bundle. A short--short crossing must lie in one cluster; the two clusters contribute at most $2\binom W2^2\leq W^4/2$. The sum is at most $7W^4/4\leq2W^4$.

A short edge cannot cross a selected bundle avoiding its cluster, since the cluster is consecutive in cap order; short edges in distinct clusters cannot alternate. Thus the listed types exhaust its unfixed neighbours, except for dummy edges, which are noncrossing with it by Lemma~\ref{lem:dummy}. The remaining fixed neighbours give unary costs. Removing internal bundle crossings leaves exactly the source graph already identified.
\end{proof}

\subsection{A computable core baseline}
Delete $Q_{\rm loc}$ from~\eqref{eq:core-energy}. Fix the palette, set all selected colours to zero, independently minimize each short unary, and insert the dummy optimum~\eqref{eq:dummy}. Call this integer $C_0$, and define
\begin{equation}\label{eq:core-baseline}
K_C=C_0-\Delta|E(G)|+\sum_{e\in\mathcal S}\min\{0,r_e\}.
\end{equation}
Crossing tests, finite sums and minima of two integers compute $K_C$
directly from the drawing.

Let $\Phi_{\rm main}(x)$ count the monochromatic crossings between selected
edges belonging to different source classes, after deleting $Q_{\rm loc}$.
For a fixed selected colouring $x$, let $\Psi(x)$ be the minimum of the
energy after deleting $Q_{\rm loc}$ over the short and dummy variables,
with the palette and fans fixed. The exact identity is
\begin{align}
\Psi(x)
&=C_0-\Delta|E(G)|+\Phi_{\rm main}(x)+\sum_{e\in\mathcal S}r_ex_e\notag\\
&=K_C+\Phi_{\rm main}(x)
+\sum_{e\in\mathcal S}\bigl(r_ex_e-\min\{0,r_e\}\bigr).
\label{eq:core-identity}
\end{align}
Indeed, the short minima and dummy value are independent of $x$ after deletion. All remaining source-dependent fixed-neighbour costs have coefficient $r_e$. The term $\Delta|E(G)|$ is the main crossing count of the all-zero selected colouring used to define the baseline.

\begin{proposition}[Core encoding]\label{prop:core-encoding}
The baseline $K_C$ is computable in polynomial time, and
\begin{equation}\label{eq:core-bound}
K_C+\Delta g(G)\leq M_C\leq
K_C+\Delta g(G)+10BSW^2+2SW^4.
\end{equation}
\end{proposition}
\begin{proof}
Every core minimizer has the palette by Lemma~\ref{lem:palette}. Deleting nonnegative local terms and applying~\eqref{eq:core-identity} gives the lower bound, since $\Phi_{\rm main}(x)\geq\Delta g(G)$ and each residual excess is nonnegative.

For the upper bound take a class-uniform selected colouring attaining $\Delta g(G)$, independently minimize every short unary, and choose an optimal dummy colouring. These choices coexist in one complete core colouring. Each selected unary exceeds its separate minimum by at most $|r_e|\leq10B$, and $|\mathcal S|=SW^2$. Restoring all local constraints costs at most $2SW^4$ by Lemma~\ref{lem:local-count}. The formula~\eqref{eq:core-baseline} proves computability with polynomially many exact operations.
\end{proof}

\begin{proposition}[Program count]\label{prop:program-count}
The number of required program points satisfies
\begin{equation}\label{eq:program-count}
s\leq\frac{16S\Lambda_{\rm pal}}{B}+\frac{8SoW}{B}+10S
\leq2^{80}(S+1)^4H^3.
\end{equation}
In particular $s\leq o$. All multiplicities are computable before final general-position insertion.
\end{proposition}
\begin{proof}
Each target in~\eqref{eq:ordinary-counts} is at most $\Lambda_{\rm pal}+ST+4o^2<2\Lambda_{\rm pal}$. A weight is at least $B/4$, so each endpoint uses at most $8\Lambda_{\rm pal}/B+1$ ordinary programs. Lemma~\ref{lem:rank-correction} uses at most $4oW/B+4$ rank programs per endpoint. Multiplication by $2S$ gives the first bound.

For $d=S+1$, $o=514SH^4<2^{10}dH^4$ and
$\Lambda_{\rm pal}=2^{26}do^2<2^{46}d^3H^8$.
Thus the first two terms are below $2^{50}d^4H^3$ and $2^{13}d^2H^3$, respectively, and $10S$ is smaller still. Their sum is below $2^{80}d^4H^3$. Since $H=2^{4200}d^{14}$, this is smaller than $514SH^4=o$ for $S\geq1$. Each step uses only the individually known response weights; placing the resulting number of copies in their open neighbourhoods preserves those weights.
\end{proof}

\section{The complete drawing and the decision reduction}\label{sec:reduction}

We use the final drawing from Proposition~\ref{prop:realization}.
For distinct crossing-graph vertex classes $A,B$, let $\Phi_{AB}$ count
their monochromatic constraints; $\Phi_A$ counts those within $A$.

\begin{figure}[!hb]
\centering
\includegraphics{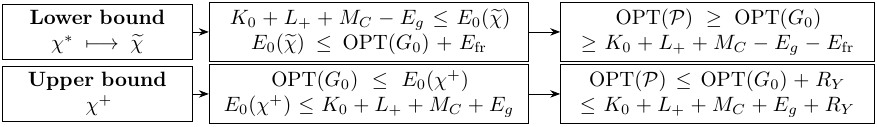}
\caption{Two-sided bounds for $G_0$ (with objective $E_0$) and $\mathcal P$.
For the lower bound, only fans of a free optimum $\chi^*$ are frozen;
Theorem~\ref{thm:grid}(3) applies to the unchanged localized grid support.
Compatible witnesses give $\chi^+$ for the upper bound, with restoration
of $Y$ costing at most $R_Y$.}
\label{fig:decoupling}
\end{figure}

\subsection{A common computable background}
Let $G_0$ be the crossing graph with $Y$ deleted. Its remaining constraints are precisely
\[
X\!:\!X,\quad X\!:\!T_f,\quad T_f\!:\!T_f,\quad C\!:\!C,\quad C\!:\!T_f,\quad I\!:\!I.
\]
Write $c^+$ for the canonical grid colouring and $t^+$ for the independently computed program-fan responses, with colour zero at a zero field. Define
\begin{align}
E_X(x;t^+)&=\Phi_X(x)+\Phi_{XT_f}(x,t^+),\notag\\
d_v&=E_X(c^+\oplus\{v\};t^+)-E_X(c^+;t^+),\notag\\
L_+&=E_X(c^+;t^+)+\sum_{v\in X}\min\{0,d_v\},\notag\\
K_0&=Z(B)+\Phi_{T_f}(t^+).
\label{eq:background}
\end{align}
The fan--fan term is evaluated after insertion, including every crossing that changed during program placement. Evaluating the canonical colours and their single-flip differences determines $K_0$ and $L_+$ explicitly.

\begin{proposition}[Decoupling with the grid fixed]\label{prop:decoupling}
Suppose the hypotheses of Theorem~\ref{thm:grid} and Proposition~\ref{prop:freezing} hold with grid dimensions $B,o$, grid size $N_g=Bo$,
\[
\rho=1/B+1/o,\quad a=\alpha+\rho\leq2^{-4096},\quad
Q=12s/o,\quad b=104a+\beta\leq1/100.
\]
Set $E_g=2^{32}N_g^2a^3$ and $E_{\rm fr}=2Q(b+\rho)bN_g^2$. Then
\begin{align}
K_0+L_++M_C-E_{\rm fr}-E_g
&\leq\OPT(G_0)\leq K_0+L_++M_C+E_g,\label{eq:decoupling-zero}\\
K_0+L_++M_C-E_{\rm fr}-E_g
&\leq\OPT(\mathcal P)\leq K_0+L_++M_C+E_g+R_Y.
\label{eq:decoupling-full}
\end{align}
\end{proposition}
\begin{proof}
Figure~\ref{fig:decoupling} summarizes the two bounds.
Take an optimum of $G_0$. Conditional on every other colour, its grid minimizes a unary-perturbed comparison energy with unary differences bounded by $\alpha N_g$. Theorem~\ref{thm:grid} supplies a canonical phase and a localized flip set of size at most $52aN_g$. If necessary complement every colour in the whole graph to make the phase positive; this preserves every monochromatic count.

Keep this grid unchanged. By Lemma~\ref{lem:soft-fields}, the probe soft-count hypothesis holds. Proposition~\ref{prop:freezing} changes only the program fans to $t^+$, with cost increase at most $E_{\rm fr}$. The grid retains its localized flip set, so the localized lower estimate in Theorem~\ref{thm:grid}, now for the frozen unary energy in~\eqref{eq:background}, gives at least $L_+-E_g$. The current core colouring costs at least $M_C$, irrespective of whether its palette was forced before the fans were frozen. The frame class costs at least $Z(B)$, and the fan--fan cost is exactly $\Phi_{T_f}(t^+)$. These terms exhaust $G_0$. Thus the altered colouring has cost at least $K_0+L_++M_C-E_g$ but at most $\OPT(G_0)+E_{\rm fr}$, proving the lower bound.

For the upper bound, use the same $t^+$, the grid colouring
$c^+\oplus\{v:d_v<0\}$, a core minimizer, and a convex frame optimum. The prescribed-phase upper witness in Theorem~\ref{thm:grid} has grid cost at most $L_++E_g$. With the fans fixed, the grid, core and frame classes have no crossing
constraints between them, so these colourings combine into a witness
for the upper bound in~\eqref{eq:decoupling-zero}. Restoring $Y$
via~\eqref{eq:nonfan-delete} then gives~\eqref{eq:decoupling-full}.
\end{proof}

\subsection{Parameter estimates}
Use the parameters in~\eqref{eq:parameters}. Take
\begin{equation}\label{eq:field-scales}
\alpha=\beta=\frac{2^{84}d^4}{H},\qquad
a=\alpha+\rho\leq\frac{2^{85}d^4}{H},\qquad
b=104a+\beta.
\end{equation}
Proposition~\ref{prop:program-count} gives $s\leq o$ and the convenient bounds
\begin{equation}\label{eq:all-sizes}
N_g\leq2^{10}dH^9,\qquad p=Bs\leq2^{80}d^4H^8,\qquad
m_C\leq2^{19}d^2H^8,\qquad h\leq2os\leq2^{91}d^5H^7.
\end{equation}
Also $B\geq2o$, $B\leq o^2/1024$, $o\geq1000$, $8\mid B$, and $W/o\leq1/514$. These verify the geometric and core hypotheses. Each follows directly by substituting $W=H^4$ and $B=H^5$ and using $H=2^{4200}d^{14}$, $d\geq1$, $S\geq1$.

In $G_0$, a grid vertex has at most $p$ outside neighbours: it has no neighbours in $C$ or $I$. A program fan has at most $p+m_C$ neighbours outside $X$, since $I$ is isolated. Thus
\[
p/N_g=s/o\leq\alpha,\qquad
(p+m_C)/N_g\leq\beta,\qquad
Q=12s/o\leq\alpha\leq a.
\]
For the second inequality, $m_C/N_g<o/(2B)$ and the spare factor in $2^{84}$ dominates this term together with $s/o$. The mesh satisfies $\rho\leq\alpha$. Finally the scale of $H$ gives $a<2^{-4096}$ and $b<1/100$. Thus the drawing satisfies the hypotheses of Proposition~\ref{prop:decoupling}.

Substitution in the error bounds gives
\begin{equation}\label{eq:individual-errors}
E_g\leq2^{307}d^{14}H^{15},\qquad
E_{\rm fr}\leq2^{290}d^{14}H^{15},\qquad
R_Y\leq2^{182}d^{10}H^{15}.
\end{equation}
For the freezing term, $Q\leq a$, $b\leq105a$, and $b+\rho\leq106a$, whence
$E_{\rm fr}\leq2\cdot105\cdot106N_g^2a^3<2^{15}N_g^2a^3$.
For $R_Y$, use $p+m_C\leq2^{81}d^4H^8$ and~\eqref{eq:all-sizes}. The resulting product is at most $2^{172}d^9H^{15}$, and $h^2/2\leq2^{181}d^{10}H^{14}$, both covered by the stated bound. Consequently
\begin{equation}\label{eq:algebraic-error}
E_{\rm fr}+2E_g+R_Y<2^{310}d^{14}H^{15}.
\end{equation}

\subsection{Quantitative reduction and threshold}
\begin{theorem}[Bounds for the target optimum]\label{thm:quantitative-reduction}
Given a represented permutation graph $G$ with $n\geq1$, the construction computes a general-position rational point set $\mathcal P$, a baseline $K_{\rm base}$, and a positive integer $\Delta=q^2H^{16}$, with explicit nonnegative errors such that
\begin{equation}\label{eq:final-window}
K_{\rm base}+\Delta g(G)-E_-
\leq\OPT(\mathcal P)\leq
K_{\rm base}+\Delta g(G)+E_+,\qquad
4(E_-+E_+)<\Delta.
\end{equation}
All output data and threshold constants have polynomial binary length and are computable in deterministic polynomial time. The points can instead be returned with integer coordinates.
\end{theorem}
\begin{proof}
Put
\begin{align}
K_{\rm base}&=K_0+L_++K_C,\notag\\
E_-&=E_{\rm fr}+E_g,\notag\\
E_+&=E_g+R_Y+10BSW^2+2SW^4.
\label{eq:final-errors}
\end{align}
The core bound~\eqref{eq:core-bound} and Proposition~\ref{prop:decoupling} give the two-sided interval.

Dividing the total error by $\Delta=q^2H^{16}$ gives the following estimates.
\begin{align*}
\frac{E_{\rm fr}+2E_g+R_Y}{\Delta}
&<\frac{2^{310}d^{14}}{q^2H}
=\frac{2^{-3890}}{q^2}<2^{-40},\\
\frac{10BSW^2}{\Delta}&=\frac{10n}{qH^3}<2^{-40},\\
\frac{2SW^4}{\Delta}&=\frac{2n}{q}<\frac18.
\end{align*}
The first inequality uses~\eqref{eq:algebraic-error} and $H=2^{4200}d^{14}$.
Hence the total is less than $1/8+2^{-39}<1/4$, which proves the strict error bound in~\eqref{eq:final-window}.

Evaluate~\eqref{eq:background} and~\eqref{eq:core-baseline} by exact crossing
tests on the final coordinates. Their ingredients are canonical colours,
single-flip and unary differences, and known convex optima. These finite
sums and minima of two integers compute $K_{\rm base}$ without a source
optimum.

Since $d=16n(n+1)+1\leq2^5(n+1)^2$, we have
$H\leq H_{\max}(n):=2^{4270}(n+1)^{28}$ and
$N=B+o+s\leq3H^5\leq3H_{\max}(n)^5$, a point bound of degree $140$ in $n$.
Proposition~\ref{prop:realization} proves polynomial coordinate bit length and construction time, and gives integer coordinates by a common positive scaling.
\end{proof}

For the threshold bit bound, the baseline formulas
\eqref{eq:background} and~\eqref{eq:core-baseline} combine fixed colouring
costs with at most $N^2$ single-edge unary corrections, each bounded by a
fixed multiple of $N^2$. Also $\Delta|E(G)|\leq N^4$, as it counts the
all-zero main crossings. Hence
\[
|K_{\rm base}|\leq16N^6.
\]
Since $q,H\leq N$, we have $\Delta\leq N^{18}$, and $|E(G)|\leq N^2$.

For a source MaxCut threshold $k_{\rm src}\in\Z$, put $t=|E(G)|-k_{\rm src}$ and output
\begin{equation}\label{eq:decision-threshold}
k_{\rm tgt}=\floor*{K_{\rm base}+\Delta t+\Delta/4}.
\end{equation}
In particular
\begin{equation}\label{eq:threshold-bits}
|k_{\rm tgt}|\leq32N^{20}(|k_{\rm src}|+1),
\end{equation}
so the threshold bit length is $O(\log N+\operatorname{bits}(k_{\rm src}))$. The baseline is integral; the floor makes the output threshold explicitly integral as well.

\begin{proof}[Proof of Theorem~\ref{thm:main}]
Use the NP-hard represented permutation-graph problem with threshold $k_{\rm src}$, or equivalently $g(G)\leq t$. For $n\geq1$ apply Theorem~\ref{thm:quantitative-reduction} and output~\eqref{eq:decision-threshold}. If $g(G)\leq t$, then
\[
\OPT(\mathcal P)\leq K_{\rm base}+\Delta t+E_+
<K_{\rm base}+\Delta t+\Delta/4.
\]
Since the objective is an integer, it is at most $k_{\rm tgt}$. If $g(G)>t$, integrality gives $g(G)\geq t+1$, and
\[
\OPT(\mathcal P)\geq K_{\rm base}+\Delta(t+1)-E_-
>K_{\rm base}+\Delta t+\Delta/4\geq k_{\rm tgt}.
\]

Empty source inputs can be decided directly. A fixed yes instance consists of two distinct integer points and threshold zero. A fixed no instance is the five points $(i,i^2)$, $0\leq i<5$, with threshold zero: their diagonal crossing graph is a five-cycle and cannot be properly two-coloured. If the target convention permits only nonnegative thresholds and the computed $k_{\rm tgt}$ is negative, output this fixed no instance. A source yes instance would otherwise imply the impossible inequality $0\leq\OPT(\mathcal P)\leq k_{\rm tgt}<0$.

Proposition~\ref{prop:realization} and~\eqref{eq:threshold-bits} give a deterministic polynomial-time many-one reduction. A common positive scaling gives integer coordinates. Together with Proposition~\ref{prop:np}, this proves NP-completeness and NP-hardness of exact minimization.
\end{proof}

\begin{corollary}
If collinear input points are permitted while the objective still counts only proper crossings using strict orientation tests, the decision problem remains NP-complete.
\end{corollary}
\begin{proof}
The general-position hard instances are also inputs in the larger class.
For membership in NP, use the verifier of Proposition~\ref{prop:np} with
the collinearity check omitted. Its strict orientation tests count only
proper crossing pairs, so verification remains polynomial-time.
\end{proof}

\section*{Verification and AI assistance}
GPT 6 astra (OpenAI) was used to assist with programmatic derivations and
verification during the development of this proof in 2026. Verification
includes a Lean formalization and independent exact computations; their
scope, external dependencies and limitations are documented in the
accompanying repository. The authors take responsibility for the content
of this article.

\section*{Data availability}
The Lean proof code, exact-computation scripts and verification
documentation supporting this article are publicly available in the
accompanying repository \cite{ProofRepository}. The repository describes
the scope of formal verification and its external dependencies.

\bibliographystyle{plain}
\bibliography{references}
\end{document}